\documentclass[aps,
prl,
twocolumn,
superscriptaddress
]{revtex4-1}
\pdfoutput=1
\usepackage[letterpaper,top=1.45cm,bottom=1.65cm,left=2cm,right=2cm]{geometry}
\usepackage{amssymb,amsmath,amsthm,graphicx,mathptmx,mathtools}
\usepackage{subfigure}
\usepackage{graphics, color}
\usepackage{latexsym}
\usepackage{bm}
\usepackage{epsfig}
\usepackage{multirow,tabularx}
\usepackage[colorlinks=true,linktocpage=true,linkcolor=blue,citecolor=blue]{hyperref}
\usepackage[abs]{overpic}

\usepackage{fancyhdr}
\fancypagestyle{plain}

\newtheorem{theorem}{Theorem}
\newtheorem{lemma}{Lemma}

\newcommand{\PRLsection}[1]{\emph{#1.---}}
\newcommand{\bs}[1]{\boldsymbol{#1}}
\newcommand{\bepsilon}{{\bs \epsilon}}
\newcommand{\ti}{\text{i}}
\newcolumntype{Y}{>{\centering\arraybackslash}X}
\def\ud{{\rm d}}

\begin{document}

\title{Well-posed initial value formulation of general effective field theories of gravity}

\author{Pau Figueras}%
\email{p.figueras@qmul.ac.uk}
\affiliation{School of Mathematical Sciences, Queen Mary University of London, Mile End Road, London E1 4NS, United Kingdom}%

\author{Aaron Held}
\email{aaron.held@phys.ens.fr}
\affiliation{Institut de Physique Th\'eorique Philippe Meyer, Laboratoire de Physique de l’\'Ecole normale
sup\'erieure (ENS), Universit\'e PSL, CNRS, Sorbonne Universit\'e, Universit\'e Paris Cit\'e,
F-75005 Paris, France}%

\author{\'Aron D. Kov\'acs}%
\email{a.kovacs@qmul.ac.uk}
\affiliation{School of Mathematical Sciences, Queen Mary University of London, Mile End Road, London E1 4NS, United Kingdom}%


\begin{abstract}
We provide a proof that \emph{all} polynomial higher-derivative effective field theories of vacuum gravity admit a well-posed initial value formulation when augmented by suitable regularising terms. 
These regularising terms can be obtained by field redefinitions and allow to rewrite the resulting equations of motion as a system of second-order nonlinear wave equations.
For instance, our result applies to the quadratic, cubic, and quartic truncations of the effective field theory of gravity that have previously appeared in the literature. 
The regularising terms correspond to fiducial massive modes, however, their masses can be chosen to be non-tachyonic and heavier than the cutoff scale and hence these modes should not affect the dynamics in the regime of validity of effective field theory.
Our well-posed formulation is not limited to the weakly coupled regime of these theories, is manifestly covariant and does neither require fine tuning of free parameters nor involves prescribing arbitrary equations.  
\end{abstract}


\maketitle
\thispagestyle{fancy}

\PRLsection{Introduction}%
The detections of gravitational waves produced in mergers of compact objects have offered the possibility of testing general relativity (GR) in a new highly dynamical and strong field regime of the theory \cite{LIGOScientific:2014pky,VIRGO:2014yos,KAGRA:2020tym,LIGOScientific:2016aoc}. Whilst up to now GR has passed all tests, the latter will become more stringent in the coming years as the sensitivity of the current detectors increases and new generations of detectors are incorporated to the network. 
Current tests are consistency tests which assume that the underlying theory of gravity is given by GR. Clearly, it would be desirable to also compare with predictions that account for possible deviations.

Without much reason to prefer a specific alternative theory of gravity, effective field theory (EFT) provides a systematic and agnostic framework to test for potential modifications. 
EFT is a powerful tool that is widely used in many different areas of physics (see for instance~\cite{Burgess:2020tbq}). In this approach, one parametrises the low-energy effects of the (unknown) high-energy degrees of freedom as a derivative expansion of the low-energy fields, allowing for all possible terms that are compatible with the symmetries, typically up to field redefinitions. The higher order terms in the expansion are suppressed by the ultraviolet (UV) length scale at which new physics is expected to appear and the unknown coefficients can in principle be determined (or constrained) from observation~\cite{Payne:2024yhk}. In the case of gravity, the Einstein-Hilbert Lagrangian is thought of as the leading order term in a series expansion in derivatives of the metric, which are typically packaged in scalars of the curvature tensor and its covariant derivatives.  
In cases where a UV completion of GR is known,
the derivative expansion can be calculated, from first principles, in a suitable low energy limit of the UV theory. The precise values of the EFT coefficients reflect the details of the UV completion. 

In general, the equations of motion that arise in truncations of the EFT are higher than second order and it is not immediately obvious how to find a well-posed initial value formulation.
The latter, however, is crucial to extract strong-field predictions of the theory.
Two generic but approximate approaches have been proposed in the literature.
In a first approach, commonly referred to as the `order-by-order' approach, one bypasses any well-posedness issues by iterating in a perturbative expansion at the level of the equations of motion~\cite{Witek:2018dmd,Okounkova:2019dfo,GalvezGhersi:2021sxs}. It turns out that the accuracy of the order-by-order approach does not only depend on the smallness of the EFT couplings but also on secular effects which accumulate over time~\cite{Okounkova:2019zjf,Okounkova:2020rqw}.
In a second approach, inspired by the M\"{u}ller-Israel-Stewart formulation of relativistic viscous hydrodynamics~\cite{Muller:1967aa,Israel:1976213,Israel:1979wp} and commonly known as the `fixing of the equations' approach, one modifies the theory by introducing auxiliary variables and corresponding (ad-hoc) evolution equations, to improve the UV behavior of the theory~\cite{Cayuso:2017iqc,Allwright:2018rut}, see~\cite{Cayuso:2020lca,Cayuso:2023xbc,Lara:2024rwa,Corman:2024cdr} for respective strong-field simulations and~\cite{Gerhardinger:2022bcw,deRham:2023ngf} for related approaches.
The auxiliary variables are introduced for certain derivatives of the metric, thus reducing the order of the equations. Then their prescribed evolution equations drive the auxiliary variables towards their `physical' values on a timescale that, ideally, is much shorter than any other physical timescale in the problem. 
However, the `fixing of the equations' method still suffers from some drawbacks. For instance, it is not covariant and it is unclear what class of modifications of the equations leaves the physical content of the theory intact. (See \cite{Geroch:1995bx,Lindblom:1995gp} for the corresponding conditions in fluid dynamics.) Moreover, the `fixing of the equations' approach involves some free parameters that have to be tuned, adding to the computational costs of the simulations. 
In contrast to the above two generic but approximate approaches, for certain classes of theories, it has been possible to find a well-posed initial value formulation~\cite{Noakes:1983xd,Kovacs:2020pns,Kovacs:2020ywu,AresteSalo:2022hua,AresteSalo:2023mmd}, and subsequently extract the strong field predictions of the theory with numerical simulations~\cite{Held:2021pht,Held:2023aap,Cayuso:2023xbc}. 
In these cases, one can solve the full equations of motion without further approximation.
A recent comparison~\cite{Corman:2024cdr} of the generic but approximate approaches to the full evolution of shift-symmetric Einstein-scalar-Gauss-Bonnet gravity quantifies potential sources of error.
Clearly, it is highly desirable to extend the class of theories for which one can evolve the full equations of motion, i.e., for which there exists a well-posed formulation without further approximation.

In this letter we present a systematic approach to `regularise' any polynomial truncation of the EFT such that the full equations of motion are well-posed. We do so by means of suitable field redefinitions and prove that the resulting initial value problem is well posed. 
Our approach does neither require the ad-hoc introduction of arbitrary equations nor fine-tuning of free parameters. Moreover, it is manifestly covariant. The equations of motion propagate a fiducial tower of massive modes, similar to the case of well-posed evolution in quadratic gravity~\cite{Stelle:1976gc,Noakes:1983xd}. At the end of this letter we argue that the masses of these modes can always be chosen to be heavier than the cutoff scale.

We use the notation and conventions as in Wald's book \cite{Wald:1984rg}. Latin letters $a,b,c,\ldots$ denote abstract indices while we use Greek letters $\mu,\nu,\ldots$ to denote spacetime coordinate indices and they run from 0 to $D-1$ in a $D$-dimensional spacetime; Latin letters $i,j,\ldots$ denote indices along a spacelike hypersurface. We use geometric units $G=c=1$ and the mostly-plus sign convention for the spacetime metric. 

\PRLsection{Preliminaries}%
We consider an EFT of vacuum gravity with a Lagrangian of the form
\begin{equation}
    \mathbf{L} = \bepsilon\,\mathcal{L}\,,\quad \mathcal{L}=\Big[R+\sum_{m\geq 2}\ell^m \mathcal{L}_m\Big]\,,
    \label{eq:lagrangian}
\end{equation}
where $\bepsilon$ is the volume form, $\mathcal{L}_m$ is a local covariant scalar of dimension $m+2$ depending only on the metric and its derivatives (and $\bepsilon$ in parity-violating theories), and $\ell$ is a length scale associated with UV physics \footnote{Note that we have set the cosmological constant to zero and we assumed that all higher derivative terms are suppressed by a single UV scale for simplicity.}. The initial value problem for generic theories of the form \eqref{eq:lagrangian} is not expected to be locally well-posed. However, we will show that we can make progress within the regime of validity of EFT. 

Let $L$ be the shortest characteristic length scale associated to a solution of \eqref{eq:lagrangian}. Such a solution is expected to be accurately described by EFT when $\ell/L\ll 1$ \footnote{For instance, in a black hole binary $L$ would correspond to the size of the smallest black hole.}. At low energies, we may truncate the infinite series \eqref{eq:lagrangian} at some finite order $N$ and work with the truncated theory, $\mathcal{L}_N^{(tr)}$. Local well-posedness of the Cauchy problem for a system of PDEs {\it generically} depends on the highest order derivatives in the system. However, as we shall see, for {\it some} PDE systems local well-posedness may also be sensitive to {\it some} (but not all) of the lower order derivative terms. Hence, to solve the initial value problem for the truncated theory we employ the following strategy. We will show that by adding a sum of suitable operators $\mathcal{L}_{\text{reg},m}$ to the truncated Lagrangian (with $m\leq m_{\text{max}}$ and a sufficiently large $m_{\text{max}}>N$), the modified (augmented) theory admits a locally well-posed initial value problem. At the same time, it is expected that the augmented theory makes the same physical predictions as the original truncated theory up to order $N$.

In particular, we consider adding to $\mathcal{L}_N^{(tr)}$ regularising Lagrangians of the form
\begin{equation}
\mathcal{L}_{\text{reg}}\equiv \sum\limits_{ k =0}^n\mathcal{L}_{\text{reg},k} = \sum\limits_{k=0}^n\ell^{2k+2}\left(\alpha_k \,R^{ab}\Box^k R_{ab}-\beta_k \,R\Box^k R \right) \label{eq:reg_lagr}
\end{equation}
for some $n$ and where $\alpha_k$, $\beta_k$ are constants. These terms in the Lagrangian can be generated by field redefinitions proportional to (derivatives of) the zeroth order equations of motion \footnote{
It is also possible to regularise the EFT at order $N$ without performing any field redefinitions. In this case, the regularising terms up to $k\leq N-2$ are already contained in the respective EFT truncation. The regularising terms with $N-2< k\leq N$ can be added at will since they are of yet higher order in the EFT and thus should not affect corrections obtained consistently at the current order. In particular, we may choose $2\alpha_n=\beta_n$ to ensure well-posedness.
}. 
Therefore, we will consider the initial value problem in the augmented theory
\begin{equation}
    \mathcal{L}_{\text{aug}}=\mathcal{L}_N^{(\text{tr})}+\mathcal{L}_{\text{reg}}\,. \label{eq:aug_lagr}
\end{equation}

In the next section we will show that well-posedness of the Cauchy problem for the augmented theory will only depend on the highest order term in the regularising Lagrangian, i.e., the piece corresponding to $k=n$ in \eqref{eq:reg_lagr}, provided that $n$ is large enough. The lower order terms in the regularising Lagrangian are included only to have better control over the particle content of the augmented theory \footnote{Note that in the EFT regime, we can make field redefinitions recursively to ensure that the derivative structure of the theory up to order $N$ is the same as in the original truncated theory.}. In the Supplemental Material we illustrate these ideas in the case of the toy problem $\Box^n u = F$.

\PRLsection{Main results}
In this section we show that (for large enough $n$) the equation of motion that result from varying the augmented Lagrangian, $\mathbf{L}_\text{aug}=\bepsilon\,\mathcal{L}_\text{aug}$, with respect to the metric can be written as a system of second order nonlinear wave equations, and hence local well-posedness of the Cauchy problem is manifest. This is our main result and it is expressed in Theorem \ref{thm:short}; the intermediate technical results are relegated to the Supplemental Material. 

We begin by introducing a set of tensor fields $R^{(p,q)}$ and $G^{(p,q)}$, which are  of type $(0,p+2)$, ${\bar R}^{(p,q)}$, which are of type $(0,p)$, and  $W^{(k)}$, which are of type $(0,k+4)$, defined as follows
\begin{align}
R^{(p,q)}_{abc_1\ldots c_p} &\equiv \nabla_{c_1}\ldots \nabla_{c_p} \Box^q R_{ab}\,, \\
G^{(p,q)}_{abc_1\ldots c_p} &\equiv \nabla_{c_1}\ldots \nabla_{c_p} \Box^q G_{ab}\,, \\
{\bar R}^{(p,q)}_{c_1\ldots c_p} &\equiv \nabla_{c_1}\ldots \nabla_{c_p} \Box^q R\,, \\
W^{(k)}_{abcd e_1\ldots e_k} &\equiv \nabla_{e_1}\ldots \nabla_{e_k} W_{abcd}\,,
\end{align}
where $R_{ab}$, $G_{ab}$ and $W_{abcd}$ are the Ricci, Einstein and Weyl tensors respectively associated to the spacetime metric $g_{ab}$ (and $R$ the Ricci scalar). Our main result can be stated as follows:
\begin{theorem}\label{thm:short}
    Assume that there exists a non-negative integer $n$ such that the equation of motion ${\cal E}^{ab}_N$ of the truncated theory $\mathcal{L}_N^{(\text{tr})}$ can be written as
    \begin{equation}
        {\cal E}^{ab}_N=F^{ab}_N \,,\label{eq:eoms_Ltr}
    \end{equation}
    where $F^{ab}_N$ is a sum of monomials such that each monomial is a product of an appropriate power of $\ell$, a coupling constant and contractions of the metric, the volume form (in case of a parity-violating theory) and a subset of the following tensor fields
        \begin{itemize}
            \item $W^{(k)}$ with $0\leq k \leq n$,
            \item ${R}^{(p,q)}$ with $p+q\leq n+1$ and $q\leq n-1$.
        \end{itemize}
    Suppose further that we employ the regularising theory with $\alpha_n=2\beta_n$. Then the augmented theory $\mathcal{L}_{\text{aug}}$ admits a well-posed initial value formulation.
\end{theorem}

Note that the assumptions of Theorem \ref{thm:short} are not too restrictive since the equations of motion of any theory with a Lagrangian of the form \eqref{eq:lagrangian} can be written as in  \eqref{eq:eoms_Ltr} provided that $n$ is sufficiently large. For example, in $D=4$ spacetime dimensions Theorem \ref{thm:short} applies to the most general $6$-derivative vacuum EFT after field redefinitions (see e.g., equation (2.5) of \cite{Endlich:2017tqa}) with $n=1$ and it applies to the $8$-derivative theory after field redefinitions (e.g., equation (1.1) of \cite{Endlich:2017tqa}) with $n=2$. In $D>4$ dimensions Theorem \ref{thm:short} applies to the most general $4$-derivative vacuum EFT after field redefinitions (Einstein-Gauss-Bonnet theory) with $n=0$. 

\begin{proof}

\noindent
{\bf Geometric identities.} We start by stating two sets of useful geometric identities involving the tensor fields $W^{(k)}$ and $R^{(p,q)}$. The first set of identities are wave equations of the form
   \begin{align}
       \Box W^{(k)}&=F^{(k)}_W \,,\label{eq:wave_W}\\
       \Box R^{(p,q)}&=R^{(p,q+1)}+F^{(p,q)}_R \,,\label{eq:wave_R}
   \end{align}
such that $F_{(k)}^W$ is a type $(0,k)$ tensor-valued polynomial in the variables $W^{(l)}$ with $0\leq l\leq k$, and $R^{(s,0)}$ with $0\leq s\leq k+2$; $F_{(p,q)}^R$ is a type $(0,p+2)$ tensor-valued polynomial in the variables $W^{(l)}$ with $0\leq l\leq p-1$, and $R^{(s,t)}$ with $s+t\leq p+q$, $t\leq q$. The second set of identities can be stated as follows. Let $n$ be a positive integer; then 
   \begin{equation}
       \nabla^bG_{ab}^{(0,n)}=I^{(n)}_a \,,  \label{eq:gradGn}
   \end{equation}
where $I^{(n)}_a$ can be expressed as a sum of monomials built from contractions of the metric and the tensor fields $W^{(k)}$ with $0\leq k \leq n-1$, $R^{(p,q)}$ with $p+q\leq n$ and $q<n$. These identities are straightforward to verify by induction and using the Bianchi and Ricci identities, see the Supplemental Material.

\noindent
{\bf Equation of motion of $\mathcal{L}_{\text{reg.}}$.} 
Consider first the theory ${\cal L}_{\text{reg.},n}$ with $\alpha_n=2\beta_n$. We claim that the equation of motion of this theory has the form
 \begin{equation}
    E^{ab}_{(n)}\equiv \Box G^{(0,n)}_{ab}+g_{ab}\nabla^c\nabla^d G^{(0,n)}_{cd}-2\nabla^c \nabla_{(a}G^{(0,n)}_{b)c}+\ldots \label{eq:eom_spec_reg}
\end{equation}
where the ellipsis stands for a sum of monomials built out of contractions of the metric and exactly two factors of the form $R^{(p,q)}$ with $p\leq 2$ and $q\leq n-1$. This statement can be verified by using induction on $n$, see Supplemental Material.

\noindent
The principal terms explicitly displayed in \eqref{eq:eom_spec_reg} have a non-diagonal structure. However, by using the identity \eqref{eq:gradGn}, it is possible to remove the second and third term on the RHS of \eqref{eq:eom_spec_reg} in favor of terms containing lower order derivatives of the curvature tensors. Equation \eqref{eq:eom_spec_reg} then reduces to (after trace-reversing)
\begin{equation}
        \Box R_{ab}^{(0,n)}=F_{ab} \label{eq:boxR_n}\,,
\end{equation}
where $F_{ab}$ is a sum of monomials built out of factors of the metric, $W^{(l)}$ with $0\leq l \leq n$, and $R^{(p,q)}$ with $p+q\leq n+1$ and $q\leq n-1$.

\noindent
Consider now the theory ${\cal L}_{\text{reg.},k}$ with $\alpha_k=0$ and $\beta_k\neq 0$. One can verify by induction on $k$ (see Supplemental Material) that the equation of motion of this theory can be written as
 \begin{equation}
    {\bar E}^{ab}_{(k)}\equiv 2\nabla^a\nabla^b {\bar R}^{(0,k)}-2\Box {\bar R}^{(0,k)}g^{ab}+\ldots \label{eq:eom_scalar_reg}
\end{equation}
where the ellipsis stands for a sum of monomials such that each monomial contains (in addition to contractions w.r.t. the metric) exactly two factors of the form ${\bar R}^{(p,q)}$ with $p\leq 2$ and $q\leq k-1$.
 
\noindent
Next, we observe that the equation of motion of the regularising Lagrangian \eqref{eq:reg_lagr} with $\alpha_n=2\beta_n$ (and $\alpha_k$, $\beta_k$ with $k<n$ left unspecified) can be written as a linear combination of the equations $E^{ab}_{(k)}$ with $k\leq n$ and ${\bar E}^{ab}_{(k)}$ with $k\leq n-1$. Moreover, it follows from the structure of $E^{ab}_{(k)}$ and ${\bar E}^{ab}_{(k)}$ that the inclusion of the lower order regularising terms retains the same structure as \eqref{eq:boxR_n}. Finally, consider a theory whose equation of motion has the structure described in the theorem and modify its Lagrangian by adding \eqref{eq:reg_lagr} with $\alpha_n=2\beta_n$. Then the equation of motion of the augmented theory can be expressed in the form \eqref{eq:boxR_n} ($F_{ab}$ will also depend on $\bepsilon$ in a parity-violating theory).

\noindent
{\bf Proof of well-posedness.} Now we are in the position to consider the initial value problem for PDEs of the form \eqref{eq:boxR_n}. Note that well-posedness of the Cauchy problem for \eqref{eq:boxR_n} is not immediately obvious because PDEs of the form $\Box^n u=F$ are only weakly hyperbolic, which means that well-posedness depends on some of the lower order derivatives, see the Supplemental Material. Nevertheless, we shall show that the structure of \eqref{eq:boxR_n} described above is such that it admits a well-posed Cauchy problem. We prove this in two steps. First, we demonstrate that an equation of the form \eqref{eq:boxR_n} can be reduced to a diagonal system of second order wave equations. Second, we show that the constraints resulting from such order reduction are propagated, i.e., solutions to the order-reduced system solve \eqref{eq:boxR_n}.

\noindent
Let $\Sigma$ be a codimension 1 spacelike manifold and let us introduce coordinates $x^i$ on it. An initial data set (corresponding to $x^0=0$) consists of the $(2n+5)$-tuple $(\Sigma,\gamma_{ij},K_{ij},{\rho}_{ij}^{(0)},\ldots {\rho}_{ij}^{(2n+1)})$ where $\gamma_{ij}$ is a Riemannian metric on $\Sigma$, $K_{ij}$ and $\rho^{(m)}_{ij}$ ($0\leq m\leq 2n+1$) are symmetric tensors on $\Sigma$ corresponding to the extrinsic curvature and $(\pounds_n)^m R_{ij}$ respectively, where $\pounds_n$ denotes de Lie derivative along the vector $n^\mu$. Next, we make an arbitrary choice for the initial value of the lapse function $\alpha$ and the shift vector $\beta^i$ (which fixes the unit normal $n^\mu$ on $\Sigma$) and choose $g_{ij}$ and $\partial_0 g_{ij}$ on $\Sigma$ so that the induced metric and the extrinsic curvature of $\Sigma$ are given by $\gamma_{ij}$ and $K_{ij}$ respectively. The first order time derivatives of $\alpha$ and $\beta^i$ can be fixed by assuming that the harmonic gauge condition holds on $\Sigma$. Initial data for the Ricci tensor is prescribed by setting $R_{ij}=\rho^{(0)}_{ij}$ on $\Sigma$ and by requiring that $R_{0\mu}$ matches its coordinate expression written in terms of $\gamma_{ij}$, $K_{ij}$ (and their derivatives parallel to $\Sigma$). For later convenience, we note that the Weyl tensor can also be computed on $\Sigma$ from $g_{\mu\nu}$, $\partial_0 g_{\mu\nu}$, their derivatives parallel to $\Sigma$, and $R_{\mu\nu}$. Similarly, $(\pounds_n)^k W|_\Sigma$ can be expressed in terms of the initial values of $g_{\mu\nu}$, $\partial_0 g_{\mu\nu}$ and $(\pounds_n)^m R_{\mu\nu}$ with $m\leq k$ and their derivatives parallel to $\Sigma$. The initial values of time derivatives of $R_{\mu\nu}$ are chosen as follows. For $\mu\nu=ij$, we set $(\pounds_n)^m R_{ij}=\rho^{(m)}_{ij}$. Initial data for $(\pounds_n)^m R_{0\mu}$ can then be determined inductively. Requiring that the contracted Bianchi identity $ \nabla^\nu G_{\mu\nu}=0$ holds on $\Sigma$ uniquely fixes the first order time derivative of $R_{0\mu}$ in terms of the data for the metric and the Ricci tensor (and derivatives parallel to $\Sigma$). Similarly, the second order time derivative of $R_{0\mu}$ is obtained using the time derivative of the contracted Bianchi identity
    \begin{equation}
       \pounds_n (\nabla^\nu G_{\mu\nu})=0\,.
    \end{equation}   
Now suppose we have already obtained data for $(\pounds_n)^m R_{0\mu}$ with $m\leq 2l$ and $1\leq l\leq n$ an integer. Consider the identity (cf. \eqref{eq:gradGn}).
    \begin{equation}
        \nabla^\nu G_{\mu\nu}^{(0,l)}=I^{(l)}_\mu \,.\label{eq:dt_Ricci_fix_1}
    \end{equation}
The RHS can be written in terms of $g_{\mu\nu}$, $\partial_0 g_{\mu\nu}$, $(\pounds_n)^k R_{\mu\nu}$ with $k\leq 2l-1$ and their derivatives parallel to $\Sigma$. Hence, \eqref{eq:dt_Ricci_fix_1} algebraically fixes $(\pounds_n)^m R_{0\mu}$ on $\Sigma$ in terms of data previously obtained. For even values of $m=2l+2$, i.e., $1\leq l\leq n-1$, we can similarly fix $(\pounds_n)^m R_{0\mu}$ by enforcing
\begin{equation}
        \pounds_n\left(\nabla^\nu G_{\mu\nu}^{(0,l)}\right)=H^{(l)}_\mu\,,
\end{equation}
on $\Sigma$ where $H^{(l)}\equiv\pounds_n I^{(l)}$ is expressible in terms of $g_{\mu\nu}$, $\partial_0 g_{\mu\nu}$, $(\pounds_n)^k R_{\mu\nu}$ with $k\leq 2l$ and their derivatives parallel to $\Sigma$.

\noindent    
Next we show that a PDE of the form \eqref{eq:boxR_n} can be reformulated in terms of a system of second order  nonlinear wave equations. In particular, it can be written as
    \begin{equation}
        g^{\alpha\beta}\partial_\alpha \partial_\beta v_A=F_{A}(v,\partial v) \,,\label{eq:reduced_wave}
    \end{equation}
where $F_{A}$ is a set of (tensor-valued) polynomials of its arguments and $v$ stands for the following list of variables: $g_{\mu\nu}$, $W^{(k)}$ with $0\leq k \leq n-1$, $R^{(p,q)}$ with $p+q\leq n$ and $q\leq n$. Initial data required for \eqref{eq:reduced_wave} consists of a codimension 1 spacelike surface $\Sigma$ and initial values for the fields $v$, $\partial_0 v$. This data can be obtained from $(\gamma_{ij},K_{ij},{\rho}_{ij}^{(0)},\ldots, {\rho}_{ij}^{(2n+1)})$ and derivatives of these fields parallel to $\Sigma$ as explained above.

\noindent
In harmonic gauge the metric can be evolved using the coordinate expression for the Ricci tensor
    \begin{equation}
        g^{\alpha\beta}\partial_\alpha \partial_\beta g_{\mu\nu}=F_{\mu\nu}^{(g)}(g,\partial g, R^{(0,0)})\,. \label{eq:wave_eqs}
    \end{equation}
The wave equations for $W^{(k)}$ and for $R^{(p,q)}$ with $p\geq 1$ and $q<n$ are given by the geometric identities in \eqref{eq:wave_W}--\eqref{eq:wave_R}. These identities are now promoted to dynamical equations for $W^{(k)}$ and for $R^{(p,q)}$. Note that in the equations for $k=n-1$ and $p=n-q+1$ one needs to replace $W^{(n)}$ with $\nabla W^{(n-1)}$ and $R^{(n-q+1,q)}$ with $\nabla R^{(n-q,q)}$. For the $R^{(p,q)}$ variables with $p=0$ and $q<n$ we simply have
\begin{equation}
        \Box R^{(0,q)}=R^{(0,q+1)}\,.
\end{equation}
Finally $R^{(0,n)}$ can be evolved using \eqref{eq:boxR_n}. Rewriting all (1st order) covariant derivatives in terms of partial derivatives and the metric connection gives the desired system of diagonal wave equations which admits a well-posed Cauchy problem~\cite{Taylor91}.

\noindent
To show that solutions of the wave system reproduce solutions of \eqref{eq:boxR_n}, we have to demonstrate that the following constraints hold throughout the evolution:
    \begin{itemize}
    \item[(i)] Algebraic constraints
    \begin{align}
    {\cal R}^{(p,q)}_{\mu\nu \sigma_1\ldots \sigma_p}&\equiv{R}^{(p,q)}_{\mu\nu \sigma_1\ldots \sigma_p}- \nabla_{\sigma_1}\ldots \nabla_{\sigma_p} \Box^q R_{\mu\nu}, \\
    {\cal W}^{(k)}_{\mu\nu\rho\sigma \alpha_1\ldots \alpha_k}&\equiv  W^{(k)}_{\mu\nu\rho\sigma \alpha_1\ldots \alpha_k}- \nabla_{\alpha_1}\ldots \nabla_{\alpha_k} W_{\mu\nu\rho\sigma},
    \end{align}
    \item[(ii)] Harmonic gauge condition
    \begin{equation}
        \Gamma^\mu=0\,,
    \end{equation}
    \item[(iii)] Bianchi constraints ($0\leq l\leq n$)
    \begin{align}
        {\cal B}^{(l)}_\mu &\equiv \nabla^\nu G_{\mu\nu}^{(l)}- I_\mu^{(l)}\,, \label{eq:Bianchi}\\
        {\cal C}^{(l)}_\mu &\equiv (\pounds_n) (\nabla^\nu G_{\mu\nu}^{(l)} )-H_\mu^{(l)}.\,\label{eq:Bianchi_dot}
    \end{align}
    \end{itemize}
    These constraints satisfy linear homogeneous wave equations that are straightforward to derive, see the Supplemental Material. Now assume we set up initial data such that the constraints $\Gamma^\mu$, ${\cal W}^{(k)}$, ${\cal R}^{(p,q)}$, ${\cal B}^{(m)}$ and ${\cal C}^{(m)}$ vanish on the initial data slice. Then there exists a unique solution to the linear homogeneous system of wave equations listed above, which must be the trivial solution. Hence $\Gamma^\mu$, ${\cal W}^{(k)}$, ${\cal R}^{(p,q)}$, ${\cal B}^{(m)}$ will be zero everywhere and consequently, the derivatives of ${\cal B}^{(m)}$, i.e., ${\cal C}^{(m)}$ will also be zero everywhere. Hence the constraints are propagated and the wave equations \eqref{eq:reduced_wave} solve equations \eqref{eq:boxR_n}.

\end{proof}

\PRLsection{Discussion}%
We have shown how to formulate a well-posed initial value problem for any polynomial truncation of the EFT of gravity. The main idea is that by performing a field redefinition of the metric, one can modify the structure of the equation of motion in a way such that they can be cast as a system of nonlinear wave equations. In the regime of validity of EFT, this field redefinition is expected to not affect the predictions of the original theory. 

The regularising terms that we add to the original theory introduce fiducial heavy degrees of freedom ($n+1$ massive spin-$0$ and $n+1$ massive spin-$2$ particles).
We fix $\alpha_n = 2\beta_n$, and the remaining coefficients $\alpha_{k\leq n}$ and $\beta_{k<n}$ of the regularising Lagrangian determine the masses of the respective (fiducial) degrees of freedom in flat spacetime. 
The latter are determined by the roots of two polynomials of degree $n+1$ (one for the spin-$0$, one for the spin-$2$ modes) and can be chosen such as to avoid tachyonic masses, in flat spacetime and, more generally, within the EFT regime of validity. For instance, the choice $\alpha_k = 2\beta_k$ and $\alpha_k=\alpha^{2k+2} \binom{n+1}{k+1}$ leads to equal and non-tachyonic masses $m^2 = (\alpha\ell)^{-2}$. For a natural value of $\alpha$, 
$(mL)^{-1}\ll (m\ell)^{-1}\sim 1$ (recall that $L$ is the shortest characteristic length scale of the low-energy physics) and 
hence, the flat space masses will dominate over any curvature-induced mass terms.
Moreover, it is thereby guaranteed that any fiducial ghost modes propagate with masses which are beyond the cutoff scale. We therefore expect that any potential instabilities, ghost-like or tachyonic, can only occur once the evolution surpasses the EFT regime of validity.
These expectations are in line with the behaviour observed in nonlinear simulations of Quadratic Gravity~\cite{Held:2023aap} but will have to be confirmed in concrete implementation. In a forthcoming publication we will show how to make such a concrete implementation in the cases of the cubic and quartic EFTs of gravity, and extract the waveforms produced in black hole binary mergers in these theories. 

A related open question is how the regularising terms affect the long time behaviour of the solutions. It would be very interesting to study this problem rigorously, even in the context of simpler EFTs as in \cite{Reall:2021ebq}.

The augmented theory requires the prescription of more initial data than in the truncated theory. One might wonder how we obtain initial data for the extra degrees of freedom. One possibility is to apply a perturbative reduction of order procedure \cite{Parker:1993dk,Flanagan:1996gw} on the initial data slice to express initial data for higher order derivatives in terms of data for $(\gamma_{ij},K_{ij})$. This approach would solve the initial value constraints up to an error term of order $\ell^{2n+2}$. It would be desirable, however, to improve on the order reduction method so as to avoid significant constraint violations in numerical simulations. We plan to investigate this problem in the future.

We expect that our approach straightforwardly generalizes to EFTs of matter fields with or without coupling to gravity. In that case one also needs to perform field redefinitions of the matter fields (in addition to those of the metric). In particular, when the matter is a fluid, our approach naturally generalizes BDNK theory \cite{Bemfica:2017wps,Bemfica:2019knx,Kovtun:2019hdm} to dissipative fluids of arbitrary order. We will leave the details for future work.

\PRLsection{Acknowledgements}%
We would like to thank Ramiro Cayuso, Harvey Reall, Bolys Sabitbek and Bob Wald for discussions. PF and ADK are supported by the STFC Consolidated Grant ST/X000931/1. AH thanks Queen Mary University of London for hospitality and the Royal Society for support as part of the Newton Alumni Grant AL\textbackslash221030.

\bibliography{refs}

\begin{thebibliography}{47}%
\makeatletter
\providecommand \@ifxundefined [1]{%
 \@ifx{#1\undefined}
}%
\providecommand \@ifnum [1]{%
 \ifnum #1\expandafter \@firstoftwo
 \else \expandafter \@secondoftwo
 \fi
}%
\providecommand \@ifx [1]{%
 \ifx #1\expandafter \@firstoftwo
 \else \expandafter \@secondoftwo
 \fi
}%
\providecommand \natexlab [1]{#1}%
\providecommand \enquote  [1]{``#1''}%
\providecommand \bibnamefont  [1]{#1}%
\providecommand \bibfnamefont [1]{#1}%
\providecommand \citenamefont [1]{#1}%
\providecommand \href@noop [0]{\@secondoftwo}%
\providecommand \href [0]{\begingroup \@sanitize@url \@href}%
\providecommand \@href[1]{\@@startlink{#1}\@@href}%
\providecommand \@@href[1]{\endgroup#1\@@endlink}%
\providecommand \@sanitize@url [0]{\catcode `\\12\catcode `\$12\catcode
  `\&12\catcode `\#12\catcode `\^12\catcode `\_12\catcode `\%12\relax}%
\providecommand \@@startlink[1]{}%
\providecommand \@@endlink[0]{}%
\providecommand \url  [0]{\begingroup\@sanitize@url \@url }%
\providecommand \@url [1]{\endgroup\@href {#1}{\urlprefix }}%
\providecommand \urlprefix  [0]{URL }%
\providecommand \Eprint [0]{\href }%
\providecommand \doibase [0]{http://dx.doi.org/}%
\providecommand \selectlanguage [0]{\@gobble}%
\providecommand \bibinfo  [0]{\@secondoftwo}%
\providecommand \bibfield  [0]{\@secondoftwo}%
\providecommand \translation [1]{[#1]}%
\providecommand \BibitemOpen [0]{}%
\providecommand \bibitemStop [0]{}%
\providecommand \bibitemNoStop [0]{.\EOS\space}%
\providecommand \EOS [0]{\spacefactor3000\relax}%
\providecommand \BibitemShut  [1]{\csname bibitem#1\endcsname}%
\let\auto@bib@innerbib\@empty
\bibitem [{\citenamefont {Aasi}\ \emph {et~al.}(2015)\citenamefont {Aasi} \emph
  {et~al.}}]{LIGOScientific:2014pky}%
  \BibitemOpen
  \bibfield  {author} {\bibinfo {author} {\bibfnamefont {J.}~\bibnamefont
  {Aasi}} \emph {et~al.} (\bibinfo {collaboration} {LIGO Scientific}),\ }\href
  {\doibase 10.1088/0264-9381/32/7/074001} {\bibfield  {journal} {\bibinfo
  {journal} {Class. Quant. Grav.}\ }\textbf {\bibinfo {volume} {32}},\ \bibinfo
  {pages} {074001} (\bibinfo {year} {2015})},\ \Eprint
  {http://arxiv.org/abs/1411.4547} {arXiv:1411.4547 [gr-qc]} \BibitemShut
  {NoStop}%
\bibitem [{\citenamefont {Acernese}\ \emph {et~al.}(2015)\citenamefont
  {Acernese} \emph {et~al.}}]{VIRGO:2014yos}%
  \BibitemOpen
  \bibfield  {author} {\bibinfo {author} {\bibfnamefont {F.}~\bibnamefont
  {Acernese}} \emph {et~al.} (\bibinfo {collaboration} {VIRGO}),\ }\href
  {\doibase 10.1088/0264-9381/32/2/024001} {\bibfield  {journal} {\bibinfo
  {journal} {Class. Quant. Grav.}\ }\textbf {\bibinfo {volume} {32}},\ \bibinfo
  {pages} {024001} (\bibinfo {year} {2015})},\ \Eprint
  {http://arxiv.org/abs/1408.3978} {arXiv:1408.3978 [gr-qc]} \BibitemShut
  {NoStop}%
\bibitem [{\citenamefont {Akutsu}\ \emph {et~al.}(2021)\citenamefont {Akutsu}
  \emph {et~al.}}]{KAGRA:2020tym}%
  \BibitemOpen
  \bibfield  {author} {\bibinfo {author} {\bibfnamefont {T.}~\bibnamefont
  {Akutsu}} \emph {et~al.} (\bibinfo {collaboration} {KAGRA}),\ }\href
  {\doibase 10.1093/ptep/ptaa125} {\bibfield  {journal} {\bibinfo  {journal}
  {PTEP}\ }\textbf {\bibinfo {volume} {2021}},\ \bibinfo {pages} {05A101}
  (\bibinfo {year} {2021})},\ \Eprint {http://arxiv.org/abs/2005.05574}
  {arXiv:2005.05574 [physics.ins-det]} \BibitemShut {NoStop}%
\bibitem [{\citenamefont {Abbott}\ \emph {et~al.}(2016)\citenamefont {Abbott}
  \emph {et~al.}}]{LIGOScientific:2016aoc}%
  \BibitemOpen
  \bibfield  {author} {\bibinfo {author} {\bibfnamefont {B.~P.}\ \bibnamefont
  {Abbott}} \emph {et~al.} (\bibinfo {collaboration} {LIGO Scientific,
  Virgo}),\ }\href {\doibase 10.1103/PhysRevLett.116.061102} {\bibfield
  {journal} {\bibinfo  {journal} {Phys. Rev. Lett.}\ }\textbf {\bibinfo
  {volume} {116}},\ \bibinfo {pages} {061102} (\bibinfo {year} {2016})},\
  \Eprint {http://arxiv.org/abs/1602.03837} {arXiv:1602.03837 [gr-qc]}
  \BibitemShut {NoStop}%
\bibitem [{\citenamefont {Burgess}(2020)}]{Burgess:2020tbq}%
  \BibitemOpen
  \bibfield  {author} {\bibinfo {author} {\bibfnamefont {C.~P.}\ \bibnamefont
  {Burgess}},\ }\href {\doibase 10.1017/9781139048040} {\emph {\bibinfo {title}
  {{Introduction to Effective Field Theory}}}}\ (\bibinfo  {publisher}
  {Cambridge University Press},\ \bibinfo {year} {2020})\BibitemShut {NoStop}%
\bibitem [{\citenamefont {Payne}\ \emph {et~al.}(2024)\citenamefont {Payne},
  \citenamefont {Isi}, \citenamefont {Chatziioannou}, \citenamefont {Lehner},
  \citenamefont {Chen},\ and\ \citenamefont {Farr}}]{Payne:2024yhk}%
  \BibitemOpen
  \bibfield  {author} {\bibinfo {author} {\bibfnamefont {E.}~\bibnamefont
  {Payne}}, \bibinfo {author} {\bibfnamefont {M.}~\bibnamefont {Isi}}, \bibinfo
  {author} {\bibfnamefont {K.}~\bibnamefont {Chatziioannou}}, \bibinfo {author}
  {\bibfnamefont {L.}~\bibnamefont {Lehner}}, \bibinfo {author} {\bibfnamefont
  {Y.}~\bibnamefont {Chen}}, \ and\ \bibinfo {author} {\bibfnamefont {W.~M.}\
  \bibnamefont {Farr}},\ }\href@noop {} {\  (\bibinfo {year} {2024})},\ \Eprint
  {http://arxiv.org/abs/2407.07043} {arXiv:2407.07043 [gr-qc]} \BibitemShut
  {NoStop}%
\bibitem [{\citenamefont {Witek}\ \emph {et~al.}(2019)\citenamefont {Witek},
  \citenamefont {Gualtieri}, \citenamefont {Pani},\ and\ \citenamefont
  {Sotiriou}}]{Witek:2018dmd}%
  \BibitemOpen
  \bibfield  {author} {\bibinfo {author} {\bibfnamefont {H.}~\bibnamefont
  {Witek}}, \bibinfo {author} {\bibfnamefont {L.}~\bibnamefont {Gualtieri}},
  \bibinfo {author} {\bibfnamefont {P.}~\bibnamefont {Pani}}, \ and\ \bibinfo
  {author} {\bibfnamefont {T.~P.}\ \bibnamefont {Sotiriou}},\ }\href {\doibase
  10.1103/PhysRevD.99.064035} {\bibfield  {journal} {\bibinfo  {journal} {Phys.
  Rev. D}\ }\textbf {\bibinfo {volume} {99}},\ \bibinfo {pages} {064035}
  (\bibinfo {year} {2019})},\ \Eprint {http://arxiv.org/abs/1810.05177}
  {arXiv:1810.05177 [gr-qc]} \BibitemShut {NoStop}%
\bibitem [{\citenamefont {Okounkova}\ \emph {et~al.}(2019)\citenamefont
  {Okounkova}, \citenamefont {Stein}, \citenamefont {Scheel},\ and\
  \citenamefont {Teukolsky}}]{Okounkova:2019dfo}%
  \BibitemOpen
  \bibfield  {author} {\bibinfo {author} {\bibfnamefont {M.}~\bibnamefont
  {Okounkova}}, \bibinfo {author} {\bibfnamefont {L.~C.}\ \bibnamefont
  {Stein}}, \bibinfo {author} {\bibfnamefont {M.~A.}\ \bibnamefont {Scheel}}, \
  and\ \bibinfo {author} {\bibfnamefont {S.~A.}\ \bibnamefont {Teukolsky}},\
  }\href {\doibase 10.1103/PhysRevD.100.104026} {\bibfield  {journal} {\bibinfo
   {journal} {Phys. Rev. D}\ }\textbf {\bibinfo {volume} {100}},\ \bibinfo
  {pages} {104026} (\bibinfo {year} {2019})},\ \Eprint
  {http://arxiv.org/abs/1906.08789} {arXiv:1906.08789 [gr-qc]} \BibitemShut
  {NoStop}%
\bibitem [{\citenamefont {G\'alvez~Ghersi}\ and\ \citenamefont
  {Stein}(2021)}]{GalvezGhersi:2021sxs}%
  \BibitemOpen
  \bibfield  {author} {\bibinfo {author} {\bibfnamefont {J.~T.}\ \bibnamefont
  {G\'alvez~Ghersi}}\ and\ \bibinfo {author} {\bibfnamefont {L.~C.}\
  \bibnamefont {Stein}},\ }\href {\doibase 10.1103/PhysRevE.104.034219}
  {\bibfield  {journal} {\bibinfo  {journal} {Phys. Rev. E}\ }\textbf {\bibinfo
  {volume} {104}},\ \bibinfo {pages} {034219} (\bibinfo {year} {2021})},\
  \Eprint {http://arxiv.org/abs/2106.08410} {arXiv:2106.08410 [hep-th]}
  \BibitemShut {NoStop}%
\bibitem [{\citenamefont {Okounkova}\ \emph {et~al.}(2020)\citenamefont
  {Okounkova}, \citenamefont {Stein}, \citenamefont {Moxon}, \citenamefont
  {Scheel},\ and\ \citenamefont {Teukolsky}}]{Okounkova:2019zjf}%
  \BibitemOpen
  \bibfield  {author} {\bibinfo {author} {\bibfnamefont {M.}~\bibnamefont
  {Okounkova}}, \bibinfo {author} {\bibfnamefont {L.~C.}\ \bibnamefont
  {Stein}}, \bibinfo {author} {\bibfnamefont {J.}~\bibnamefont {Moxon}},
  \bibinfo {author} {\bibfnamefont {M.~A.}\ \bibnamefont {Scheel}}, \ and\
  \bibinfo {author} {\bibfnamefont {S.~A.}\ \bibnamefont {Teukolsky}},\ }\href
  {\doibase 10.1103/PhysRevD.101.104016} {\bibfield  {journal} {\bibinfo
  {journal} {Phys. Rev. D}\ }\textbf {\bibinfo {volume} {101}},\ \bibinfo
  {pages} {104016} (\bibinfo {year} {2020})},\ \Eprint
  {http://arxiv.org/abs/1911.02588} {arXiv:1911.02588 [gr-qc]} \BibitemShut
  {NoStop}%
\bibitem [{\citenamefont {Okounkova}(2020)}]{Okounkova:2020rqw}%
  \BibitemOpen
  \bibfield  {author} {\bibinfo {author} {\bibfnamefont {M.}~\bibnamefont
  {Okounkova}},\ }\href {\doibase 10.1103/PhysRevD.102.084046} {\bibfield
  {journal} {\bibinfo  {journal} {Phys. Rev. D}\ }\textbf {\bibinfo {volume}
  {102}},\ \bibinfo {pages} {084046} (\bibinfo {year} {2020})},\ \Eprint
  {http://arxiv.org/abs/2001.03571} {arXiv:2001.03571 [gr-qc]} \BibitemShut
  {NoStop}%
\bibitem [{\citenamefont {M{\"u}ller}(1967)}]{Muller:1967aa}%
  \BibitemOpen
  \bibfield  {author} {\bibinfo {author} {\bibfnamefont {I.}~\bibnamefont
  {M{\"u}ller}},\ }\href@noop {} {\bibfield  {journal} {\bibinfo  {journal}
  {Zeitschrift f{\"u}r Physik}\ }\textbf {\bibinfo {volume} {198}},\ \bibinfo
  {pages} {329} (\bibinfo {year} {1967})}\BibitemShut {NoStop}%
\bibitem [{\citenamefont {Israel}\ and\ \citenamefont
  {Stewart}(1976)}]{Israel:1976213}%
  \BibitemOpen
  \bibfield  {author} {\bibinfo {author} {\bibfnamefont {W.}~\bibnamefont
  {Israel}}\ and\ \bibinfo {author} {\bibfnamefont {J.}~\bibnamefont
  {Stewart}},\ }\href {\doibase https://doi.org/10.1016/0375-9601(76)90075-X}
  {\bibfield  {journal} {\bibinfo  {journal} {Physics Letters A}\ }\textbf
  {\bibinfo {volume} {58}},\ \bibinfo {pages} {213 } (\bibinfo {year}
  {1976})}\BibitemShut {NoStop}%
\bibitem [{\citenamefont {Israel}\ and\ \citenamefont
  {Stewart}(1979)}]{Israel:1979wp}%
  \BibitemOpen
  \bibfield  {author} {\bibinfo {author} {\bibfnamefont {W.}~\bibnamefont
  {Israel}}\ and\ \bibinfo {author} {\bibfnamefont {J.~M.}\ \bibnamefont
  {Stewart}},\ }\href {\doibase 10.1016/0003-4916(79)90130-1} {\bibfield
  {journal} {\bibinfo  {journal} {Annals Phys.}\ }\textbf {\bibinfo {volume}
  {118}},\ \bibinfo {pages} {341} (\bibinfo {year} {1979})}\BibitemShut
  {NoStop}%
\bibitem [{\citenamefont {Cayuso}\ \emph {et~al.}(2017)\citenamefont {Cayuso},
  \citenamefont {Ortiz},\ and\ \citenamefont {Lehner}}]{Cayuso:2017iqc}%
  \BibitemOpen
  \bibfield  {author} {\bibinfo {author} {\bibfnamefont {J.}~\bibnamefont
  {Cayuso}}, \bibinfo {author} {\bibfnamefont {N.}~\bibnamefont {Ortiz}}, \
  and\ \bibinfo {author} {\bibfnamefont {L.}~\bibnamefont {Lehner}},\ }\href
  {\doibase 10.1103/PhysRevD.96.084043} {\bibfield  {journal} {\bibinfo
  {journal} {Phys. Rev. D}\ }\textbf {\bibinfo {volume} {96}},\ \bibinfo
  {pages} {084043} (\bibinfo {year} {2017})},\ \Eprint
  {http://arxiv.org/abs/1706.07421} {arXiv:1706.07421 [gr-qc]} \BibitemShut
  {NoStop}%
\bibitem [{\citenamefont {Allwright}\ and\ \citenamefont
  {Lehner}(2019)}]{Allwright:2018rut}%
  \BibitemOpen
  \bibfield  {author} {\bibinfo {author} {\bibfnamefont {G.}~\bibnamefont
  {Allwright}}\ and\ \bibinfo {author} {\bibfnamefont {L.}~\bibnamefont
  {Lehner}},\ }\href {\doibase 10.1088/1361-6382/ab0ee1} {\bibfield  {journal}
  {\bibinfo  {journal} {Class. Quant. Grav.}\ }\textbf {\bibinfo {volume}
  {36}},\ \bibinfo {pages} {084001} (\bibinfo {year} {2019})},\ \Eprint
  {http://arxiv.org/abs/1808.07897} {arXiv:1808.07897 [gr-qc]} \BibitemShut
  {NoStop}%
\bibitem [{\citenamefont {Cayuso}\ and\ \citenamefont
  {Lehner}(2020)}]{Cayuso:2020lca}%
  \BibitemOpen
  \bibfield  {author} {\bibinfo {author} {\bibfnamefont {R.}~\bibnamefont
  {Cayuso}}\ and\ \bibinfo {author} {\bibfnamefont {L.}~\bibnamefont
  {Lehner}},\ }\href {\doibase 10.1103/PhysRevD.102.084008} {\bibfield
  {journal} {\bibinfo  {journal} {Phys. Rev. D}\ }\textbf {\bibinfo {volume}
  {102}},\ \bibinfo {pages} {084008} (\bibinfo {year} {2020})},\ \Eprint
  {http://arxiv.org/abs/2005.13720} {arXiv:2005.13720 [gr-qc]} \BibitemShut
  {NoStop}%
\bibitem [{\citenamefont {Cayuso}\ \emph {et~al.}(2023)\citenamefont {Cayuso},
  \citenamefont {Figueras}, \citenamefont {Fran\c{c}a},\ and\ \citenamefont
  {Lehner}}]{Cayuso:2023xbc}%
  \BibitemOpen
  \bibfield  {author} {\bibinfo {author} {\bibfnamefont {R.}~\bibnamefont
  {Cayuso}}, \bibinfo {author} {\bibfnamefont {P.}~\bibnamefont {Figueras}},
  \bibinfo {author} {\bibfnamefont {T.}~\bibnamefont {Fran\c{c}a}}, \ and\
  \bibinfo {author} {\bibfnamefont {L.}~\bibnamefont {Lehner}},\ }\href
  {\doibase 10.1103/PhysRevLett.131.111403} {\bibfield  {journal} {\bibinfo
  {journal} {Phys. Rev. Lett.}\ }\textbf {\bibinfo {volume} {131}},\ \bibinfo
  {pages} {111403} (\bibinfo {year} {2023})}\BibitemShut {NoStop}%
\bibitem [{\citenamefont {Lara}\ \emph {et~al.}(2024)\citenamefont {Lara},
  \citenamefont {Pfeiffer}, \citenamefont {Wittek}, \citenamefont {Vu},
  \citenamefont {Nelli}, \citenamefont {Carpenter}, \citenamefont {Lovelace},
  \citenamefont {Scheel},\ and\ \citenamefont {Throwe}}]{Lara:2024rwa}%
  \BibitemOpen
  \bibfield  {author} {\bibinfo {author} {\bibfnamefont {G.}~\bibnamefont
  {Lara}}, \bibinfo {author} {\bibfnamefont {H.~P.}\ \bibnamefont {Pfeiffer}},
  \bibinfo {author} {\bibfnamefont {N.~A.}\ \bibnamefont {Wittek}}, \bibinfo
  {author} {\bibfnamefont {N.~L.}\ \bibnamefont {Vu}}, \bibinfo {author}
  {\bibfnamefont {K.~C.}\ \bibnamefont {Nelli}}, \bibinfo {author}
  {\bibfnamefont {A.}~\bibnamefont {Carpenter}}, \bibinfo {author}
  {\bibfnamefont {G.}~\bibnamefont {Lovelace}}, \bibinfo {author}
  {\bibfnamefont {M.~A.}\ \bibnamefont {Scheel}}, \ and\ \bibinfo {author}
  {\bibfnamefont {W.}~\bibnamefont {Throwe}},\ }\href@noop {} {\  (\bibinfo
  {year} {2024})},\ \Eprint {http://arxiv.org/abs/2403.08705} {arXiv:2403.08705
  [gr-qc]} \BibitemShut {NoStop}%
\bibitem [{\citenamefont {Corman}\ \emph {et~al.}(2024)\citenamefont {Corman},
  \citenamefont {Lehner}, \citenamefont {East},\ and\ \citenamefont
  {Dideron}}]{Corman:2024cdr}%
  \BibitemOpen
  \bibfield  {author} {\bibinfo {author} {\bibfnamefont {M.}~\bibnamefont
  {Corman}}, \bibinfo {author} {\bibfnamefont {L.}~\bibnamefont {Lehner}},
  \bibinfo {author} {\bibfnamefont {W.~E.}\ \bibnamefont {East}}, \ and\
  \bibinfo {author} {\bibfnamefont {G.}~\bibnamefont {Dideron}},\ }\href@noop
  {} {\  (\bibinfo {year} {2024})},\ \Eprint {http://arxiv.org/abs/2405.15581}
  {arXiv:2405.15581 [gr-qc]} \BibitemShut {NoStop}%
\bibitem [{\citenamefont {Gerhardinger}\ \emph {et~al.}(2022)\citenamefont
  {Gerhardinger}, \citenamefont {Giblin}, \citenamefont {Tolley},\ and\
  \citenamefont {Trodden}}]{Gerhardinger:2022bcw}%
  \BibitemOpen
  \bibfield  {author} {\bibinfo {author} {\bibfnamefont {M.}~\bibnamefont
  {Gerhardinger}}, \bibinfo {author} {\bibfnamefont {J.~T.}\ \bibnamefont
  {Giblin}, \bibfnamefont {Jr.}}, \bibinfo {author} {\bibfnamefont {A.~J.}\
  \bibnamefont {Tolley}}, \ and\ \bibinfo {author} {\bibfnamefont
  {M.}~\bibnamefont {Trodden}},\ }\href {\doibase 10.1103/PhysRevD.106.043522}
  {\bibfield  {journal} {\bibinfo  {journal} {Phys. Rev. D}\ }\textbf {\bibinfo
  {volume} {106}},\ \bibinfo {pages} {043522} (\bibinfo {year} {2022})},\
  \Eprint {http://arxiv.org/abs/2205.05697} {arXiv:2205.05697 [hep-th]}
  \BibitemShut {NoStop}%
\bibitem [{\citenamefont {de~Rham}\ \emph {et~al.}(2023)\citenamefont
  {de~Rham}, \citenamefont {Ko\.zuszek}, \citenamefont {Tolley},\ and\
  \citenamefont {Wiseman}}]{deRham:2023ngf}%
  \BibitemOpen
  \bibfield  {author} {\bibinfo {author} {\bibfnamefont {C.}~\bibnamefont
  {de~Rham}}, \bibinfo {author} {\bibfnamefont {J.}~\bibnamefont {Ko\.zuszek}},
  \bibinfo {author} {\bibfnamefont {A.~J.}\ \bibnamefont {Tolley}}, \ and\
  \bibinfo {author} {\bibfnamefont {T.}~\bibnamefont {Wiseman}},\ }\href
  {\doibase 10.1103/PhysRevD.108.084052} {\bibfield  {journal} {\bibinfo
  {journal} {Phys. Rev. D}\ }\textbf {\bibinfo {volume} {108}},\ \bibinfo
  {pages} {084052} (\bibinfo {year} {2023})},\ \Eprint
  {http://arxiv.org/abs/2302.04876} {arXiv:2302.04876 [hep-th]} \BibitemShut
  {NoStop}%
\bibitem [{\citenamefont {Geroch}(1995)}]{Geroch:1995bx}%
  \BibitemOpen
  \bibfield  {author} {\bibinfo {author} {\bibfnamefont {R.~P.}\ \bibnamefont
  {Geroch}},\ }\href {\doibase 10.1063/1.530958} {\bibfield  {journal}
  {\bibinfo  {journal} {J. Math. Phys.}\ }\textbf {\bibinfo {volume} {36}},\
  \bibinfo {pages} {4226} (\bibinfo {year} {1995})}\BibitemShut {NoStop}%
\bibitem [{\citenamefont {Lindblom}(1996)}]{Lindblom:1995gp}%
  \BibitemOpen
  \bibfield  {author} {\bibinfo {author} {\bibfnamefont {L.}~\bibnamefont
  {Lindblom}},\ }\href {\doibase 10.1006/aphy.1996.0036} {\bibfield  {journal}
  {\bibinfo  {journal} {Annals Phys.}\ }\textbf {\bibinfo {volume} {247}},\
  \bibinfo {pages} {1} (\bibinfo {year} {1996})},\ \Eprint
  {http://arxiv.org/abs/gr-qc/9508058} {arXiv:gr-qc/9508058} \BibitemShut
  {NoStop}%
\bibitem [{\citenamefont {Noakes}(1983)}]{Noakes:1983xd}%
  \BibitemOpen
  \bibfield  {author} {\bibinfo {author} {\bibfnamefont {D.~R.}\ \bibnamefont
  {Noakes}},\ }\href {\doibase 10.1063/1.525906} {\bibfield  {journal}
  {\bibinfo  {journal} {J. Math. Phys.}\ }\textbf {\bibinfo {volume} {24}},\
  \bibinfo {pages} {1846} (\bibinfo {year} {1983})}\BibitemShut {NoStop}%
\bibitem [{\citenamefont {Kov\'acs}\ and\ \citenamefont
  {Reall}(2020{\natexlab{a}})}]{Kovacs:2020pns}%
  \BibitemOpen
  \bibfield  {author} {\bibinfo {author} {\bibfnamefont {A.~D.}\ \bibnamefont
  {Kov\'acs}}\ and\ \bibinfo {author} {\bibfnamefont {H.~S.}\ \bibnamefont
  {Reall}},\ }\href {\doibase 10.1103/PhysRevLett.124.221101} {\bibfield
  {journal} {\bibinfo  {journal} {Phys. Rev. Lett.}\ }\textbf {\bibinfo
  {volume} {124}},\ \bibinfo {pages} {221101} (\bibinfo {year}
  {2020}{\natexlab{a}})},\ \Eprint {http://arxiv.org/abs/2003.04327}
  {arXiv:2003.04327 [gr-qc]} \BibitemShut {NoStop}%
\bibitem [{\citenamefont {Kov\'acs}\ and\ \citenamefont
  {Reall}(2020{\natexlab{b}})}]{Kovacs:2020ywu}%
  \BibitemOpen
  \bibfield  {author} {\bibinfo {author} {\bibfnamefont {A.~D.}\ \bibnamefont
  {Kov\'acs}}\ and\ \bibinfo {author} {\bibfnamefont {H.~S.}\ \bibnamefont
  {Reall}},\ }\href {\doibase 10.1103/PhysRevD.101.124003} {\bibfield
  {journal} {\bibinfo  {journal} {Phys. Rev. D}\ }\textbf {\bibinfo {volume}
  {101}},\ \bibinfo {pages} {124003} (\bibinfo {year} {2020}{\natexlab{b}})},\
  \Eprint {http://arxiv.org/abs/2003.08398} {arXiv:2003.08398 [gr-qc]}
  \BibitemShut {NoStop}%
\bibitem [{\citenamefont {Arest\'e~Sal\'o}\ \emph {et~al.}(2022)\citenamefont
  {Arest\'e~Sal\'o}, \citenamefont {Clough},\ and\ \citenamefont
  {Figueras}}]{AresteSalo:2022hua}%
  \BibitemOpen
  \bibfield  {author} {\bibinfo {author} {\bibfnamefont {L.}~\bibnamefont
  {Arest\'e~Sal\'o}}, \bibinfo {author} {\bibfnamefont {K.}~\bibnamefont
  {Clough}}, \ and\ \bibinfo {author} {\bibfnamefont {P.}~\bibnamefont
  {Figueras}},\ }\href {\doibase 10.1103/PhysRevLett.129.261104} {\bibfield
  {journal} {\bibinfo  {journal} {Phys. Rev. Lett.}\ }\textbf {\bibinfo
  {volume} {129}},\ \bibinfo {pages} {261104} (\bibinfo {year} {2022})},\
  \Eprint {http://arxiv.org/abs/2208.14470} {arXiv:2208.14470 [gr-qc]}
  \BibitemShut {NoStop}%
\bibitem [{\citenamefont {Arest\'e~Sal\'o}\ \emph {et~al.}(2023)\citenamefont
  {Arest\'e~Sal\'o}, \citenamefont {Clough},\ and\ \citenamefont
  {Figueras}}]{AresteSalo:2023mmd}%
  \BibitemOpen
  \bibfield  {author} {\bibinfo {author} {\bibfnamefont {L.}~\bibnamefont
  {Arest\'e~Sal\'o}}, \bibinfo {author} {\bibfnamefont {K.}~\bibnamefont
  {Clough}}, \ and\ \bibinfo {author} {\bibfnamefont {P.}~\bibnamefont
  {Figueras}},\ }\href {\doibase 10.1103/PhysRevD.108.084018} {\bibfield
  {journal} {\bibinfo  {journal} {Phys. Rev. D}\ }\textbf {\bibinfo {volume}
  {108}},\ \bibinfo {pages} {084018} (\bibinfo {year} {2023})},\ \Eprint
  {http://arxiv.org/abs/2306.14966} {arXiv:2306.14966 [gr-qc]} \BibitemShut
  {NoStop}%
\bibitem [{\citenamefont {Held}\ and\ \citenamefont
  {Lim}(2021)}]{Held:2021pht}%
  \BibitemOpen
  \bibfield  {author} {\bibinfo {author} {\bibfnamefont {A.}~\bibnamefont
  {Held}}\ and\ \bibinfo {author} {\bibfnamefont {H.}~\bibnamefont {Lim}},\
  }\href {\doibase 10.1103/PhysRevD.104.084075} {\bibfield  {journal} {\bibinfo
   {journal} {Phys. Rev. D}\ }\textbf {\bibinfo {volume} {104}},\ \bibinfo
  {pages} {084075} (\bibinfo {year} {2021})},\ \Eprint
  {http://arxiv.org/abs/2104.04010} {arXiv:2104.04010 [gr-qc]} \BibitemShut
  {NoStop}%
\bibitem [{\citenamefont {Held}\ and\ \citenamefont
  {Lim}(2023)}]{Held:2023aap}%
  \BibitemOpen
  \bibfield  {author} {\bibinfo {author} {\bibfnamefont {A.}~\bibnamefont
  {Held}}\ and\ \bibinfo {author} {\bibfnamefont {H.}~\bibnamefont {Lim}},\
  }\href {\doibase 10.1103/PhysRevD.108.104025} {\bibfield  {journal} {\bibinfo
   {journal} {Phys. Rev. D}\ }\textbf {\bibinfo {volume} {108}},\ \bibinfo
  {pages} {104025} (\bibinfo {year} {2023})},\ \Eprint
  {http://arxiv.org/abs/2306.04725} {arXiv:2306.04725 [gr-qc]} \BibitemShut
  {NoStop}%
\bibitem [{\citenamefont {Stelle}(1977)}]{Stelle:1976gc}%
  \BibitemOpen
  \bibfield  {author} {\bibinfo {author} {\bibfnamefont {K.~S.}\ \bibnamefont
  {Stelle}},\ }\href {\doibase 10.1103/PhysRevD.16.953} {\bibfield  {journal}
  {\bibinfo  {journal} {Phys. Rev. D}\ }\textbf {\bibinfo {volume} {16}},\
  \bibinfo {pages} {953} (\bibinfo {year} {1977})}\BibitemShut {NoStop}%
\bibitem [{\citenamefont {Wald}(1984)}]{Wald:1984rg}%
  \BibitemOpen
  \bibfield  {author} {\bibinfo {author} {\bibfnamefont {R.~M.}\ \bibnamefont
  {Wald}},\ }\href {\doibase 10.7208/chicago/9780226870373.001.0001} {\emph
  {\bibinfo {title} {{General Relativity}}}}\ (\bibinfo  {publisher} {Chicago
  Univ. Pr.},\ \bibinfo {address} {Chicago, USA},\ \bibinfo {year}
  {1984})\BibitemShut {NoStop}%
\bibitem [{Note1()}]{Note1}%
  \BibitemOpen
  \bibinfo {note} {Note that we have set the cosmological constant to zero and
  we assumed that all higher derivative terms are suppressed by a single UV
  scale for simplicity.}\BibitemShut {Stop}%
\bibitem [{Note2()}]{Note2}%
  \BibitemOpen
  \bibinfo {note} {For instance, in a black hole binary $L$ would correspond to
  the size of the smallest black hole.}\BibitemShut {Stop}%
\bibitem [{Note3()}]{Note3}%
  \BibitemOpen
  \bibinfo {note} {It is also possible to regularise the EFT at order $N$
  without performing any field redefinitions. In this case, the regularising
  terms up to $k\leq N-2$ are already contained in the respective EFT
  truncation. The regularising terms with $N-2< k\leq N$ can be added at will
  since they are of yet higher order in the EFT and thus should not affect
  corrections obtained consistently at the current order. In particular, we may
  choose $2\alpha _n=\beta _n$ to ensure well-posedness.}\BibitemShut {Stop}%
\bibitem [{Note4()}]{Note4}%
  \BibitemOpen
  \bibinfo {note} {Note that in the EFT regime, we can make field redefinitions
  recursively to ensure that the derivative structure of the theory up to order
  $N$ is the same as in the original truncated theory.}\BibitemShut {Stop}%
\bibitem [{\citenamefont {Endlich}\ \emph {et~al.}(2017)\citenamefont
  {Endlich}, \citenamefont {Gorbenko}, \citenamefont {Huang},\ and\
  \citenamefont {Senatore}}]{Endlich:2017tqa}%
  \BibitemOpen
  \bibfield  {author} {\bibinfo {author} {\bibfnamefont {S.}~\bibnamefont
  {Endlich}}, \bibinfo {author} {\bibfnamefont {V.}~\bibnamefont {Gorbenko}},
  \bibinfo {author} {\bibfnamefont {J.}~\bibnamefont {Huang}}, \ and\ \bibinfo
  {author} {\bibfnamefont {L.}~\bibnamefont {Senatore}},\ }\href {\doibase
  10.1007/JHEP09(2017)122} {\bibfield  {journal} {\bibinfo  {journal} {JHEP}\
  }\textbf {\bibinfo {volume} {09}},\ \bibinfo {pages} {122} (\bibinfo {year}
  {2017})},\ \Eprint {http://arxiv.org/abs/1704.01590} {arXiv:1704.01590
  [gr-qc]} \BibitemShut {NoStop}%
\bibitem [{\citenamefont {Taylor}(1991)}]{Taylor91}%
  \BibitemOpen
  \bibfield  {author} {\bibinfo {author} {\bibfnamefont {M.~E.}\ \bibnamefont
  {Taylor}},\ }\href {\doibase 10.1007/978-1-4612-0431-2} {\emph {\bibinfo
  {title} {{Pseudodifferential Operators and Nonlinear PDE}}}}\ (\bibinfo
  {publisher} {Birkh{\"{a}}user, Boston, MA},\ \bibinfo {address} {Boston},\
  \bibinfo {year} {1991})\BibitemShut {NoStop}%
\bibitem [{\citenamefont {Reall}\ and\ \citenamefont
  {Warnick}(2022)}]{Reall:2021ebq}%
  \BibitemOpen
  \bibfield  {author} {\bibinfo {author} {\bibfnamefont {H.~S.}\ \bibnamefont
  {Reall}}\ and\ \bibinfo {author} {\bibfnamefont {C.~M.}\ \bibnamefont
  {Warnick}},\ }\href {\doibase 10.1063/5.0075455} {\bibfield  {journal}
  {\bibinfo  {journal} {J. Math. Phys.}\ }\textbf {\bibinfo {volume} {63}},\
  \bibinfo {pages} {042901} (\bibinfo {year} {2022})},\ \Eprint
  {http://arxiv.org/abs/2105.12028} {arXiv:2105.12028 [hep-th]} \BibitemShut
  {NoStop}%
\bibitem [{\citenamefont {Parker}\ and\ \citenamefont
  {Simon}(1993)}]{Parker:1993dk}%
  \BibitemOpen
  \bibfield  {author} {\bibinfo {author} {\bibfnamefont {L.}~\bibnamefont
  {Parker}}\ and\ \bibinfo {author} {\bibfnamefont {J.~Z.}\ \bibnamefont
  {Simon}},\ }\href {\doibase 10.1103/PhysRevD.47.1339} {\bibfield  {journal}
  {\bibinfo  {journal} {Phys. Rev. D}\ }\textbf {\bibinfo {volume} {47}},\
  \bibinfo {pages} {1339} (\bibinfo {year} {1993})},\ \Eprint
  {http://arxiv.org/abs/gr-qc/9211002} {arXiv:gr-qc/9211002} \BibitemShut
  {NoStop}%
\bibitem [{\citenamefont {Flanagan}\ and\ \citenamefont
  {Wald}(1996)}]{Flanagan:1996gw}%
  \BibitemOpen
  \bibfield  {author} {\bibinfo {author} {\bibfnamefont {E.~E.}\ \bibnamefont
  {Flanagan}}\ and\ \bibinfo {author} {\bibfnamefont {R.~M.}\ \bibnamefont
  {Wald}},\ }\href {\doibase 10.1103/PhysRevD.54.6233} {\bibfield  {journal}
  {\bibinfo  {journal} {Phys. Rev.}\ }\textbf {\bibinfo {volume} {D54}},\
  \bibinfo {pages} {6233} (\bibinfo {year} {1996})}\BibitemShut {NoStop}%
\bibitem [{\citenamefont {Bemfica}\ \emph {et~al.}(2018)\citenamefont
  {Bemfica}, \citenamefont {Disconzi},\ and\ \citenamefont
  {Noronha}}]{Bemfica:2017wps}%
  \BibitemOpen
  \bibfield  {author} {\bibinfo {author} {\bibfnamefont {F.~S.}\ \bibnamefont
  {Bemfica}}, \bibinfo {author} {\bibfnamefont {M.~M.}\ \bibnamefont
  {Disconzi}}, \ and\ \bibinfo {author} {\bibfnamefont {J.}~\bibnamefont
  {Noronha}},\ }\href {\doibase 10.1103/PhysRevD.98.104064} {\bibfield
  {journal} {\bibinfo  {journal} {Phys. Rev. D}\ }\textbf {\bibinfo {volume}
  {98}},\ \bibinfo {pages} {104064} (\bibinfo {year} {2018})},\ \Eprint
  {http://arxiv.org/abs/1708.06255} {arXiv:1708.06255 [gr-qc]} \BibitemShut
  {NoStop}%
\bibitem [{\citenamefont {Bemfica}\ \emph {et~al.}(2019)\citenamefont
  {Bemfica}, \citenamefont {Bemfica}, \citenamefont {Disconzi}, \citenamefont
  {Disconzi}, \citenamefont {Noronha},\ and\ \citenamefont
  {Noronha}}]{Bemfica:2019knx}%
  \BibitemOpen
  \bibfield  {author} {\bibinfo {author} {\bibfnamefont {F.~S.}\ \bibnamefont
  {Bemfica}}, \bibinfo {author} {\bibfnamefont {F.~S.}\ \bibnamefont
  {Bemfica}}, \bibinfo {author} {\bibfnamefont {M.~M.}\ \bibnamefont
  {Disconzi}}, \bibinfo {author} {\bibfnamefont {M.~M.}\ \bibnamefont
  {Disconzi}}, \bibinfo {author} {\bibfnamefont {J.}~\bibnamefont {Noronha}}, \
  and\ \bibinfo {author} {\bibfnamefont {J.}~\bibnamefont {Noronha}},\ }\href
  {\doibase 10.1103/PhysRevD.100.104020} {\bibfield  {journal} {\bibinfo
  {journal} {Phys. Rev. D}\ }\textbf {\bibinfo {volume} {100}},\ \bibinfo
  {pages} {104020} (\bibinfo {year} {2019})},\ \bibinfo {note} {[Erratum:
  Phys.Rev.D 105, 069902 (2022)]},\ \Eprint {http://arxiv.org/abs/1907.12695}
  {arXiv:1907.12695 [gr-qc]} \BibitemShut {NoStop}%
\bibitem [{\citenamefont {Kovtun}(2019)}]{Kovtun:2019hdm}%
  \BibitemOpen
  \bibfield  {author} {\bibinfo {author} {\bibfnamefont {P.}~\bibnamefont
  {Kovtun}},\ }\href {\doibase 10.1007/JHEP10(2019)034} {\bibfield  {journal}
  {\bibinfo  {journal} {JHEP}\ }\textbf {\bibinfo {volume} {10}},\ \bibinfo
  {pages} {034} (\bibinfo {year} {2019})},\ \Eprint
  {http://arxiv.org/abs/1907.08191} {arXiv:1907.08191 [hep-th]} \BibitemShut
  {NoStop}%
\bibitem [{\citenamefont {Kreiss}\ and\ \citenamefont
  {Lorenz}(1989)}]{Kreiss1989}%
  \BibitemOpen
  \bibfield  {author} {\bibinfo {author} {\bibfnamefont {H.-O.}\ \bibnamefont
  {Kreiss}}\ and\ \bibinfo {author} {\bibfnamefont {J.}~\bibnamefont
  {Lorenz}},\ }\href {\doibase 10.1137/1.9780898719130} {\emph {\bibinfo
  {title} {{Initial-boundary value problems and the Navier-Stokes
  equations}}}},\ Vol.\ \bibinfo {volume} {136}\ (\bibinfo  {publisher}
  {Society for Industrial and Applied Mathematics},\ \bibinfo {year}
  {1989})\BibitemShut {NoStop}%
\bibitem [{\citenamefont {Choquet-Bruhat}(2009)}]{Choquet-Bruhat:2009xil}%
  \BibitemOpen
  \bibfield  {author} {\bibinfo {author} {\bibfnamefont {Y.}~\bibnamefont
  {Choquet-Bruhat}},\ }\href@noop {} {\emph {\bibinfo {title} {{General
  Relativity and the Einstein Equations}}}},\ Oxford Mathematical Monographs\
  (\bibinfo  {publisher} {Oxford University Press},\ \bibinfo {address} {United
  Kingdom},\ \bibinfo {year} {2009})\BibitemShut {NoStop}%
\end{thebibliography}%
\bibliographystyle{apsrev4-1}

\clearpage


\widetext
\begin{center}
{\large Supplemental Material}
\end{center}

\setcounter{equation}{0}
\setcounter{page}{1}
\makeatletter
\renewcommand{\theequation}{S\arabic{equation}}

\section{A simple example}

In this section we illustrate some of the ideas presented in the Letter on a simple equation. Let $(t,x)$ be coordinates in $(1+1)$-dimensional Minkowski space and consider a PDE
\begin{equation}
    \Box^{n} u =F(u,\ldots, \partial^{2n-1}u)\,, \label{eq:toy_n}
\end{equation}
for a scalar field $u$ with initial conditions
\begin{equation}
     u(0,x)=f_0(x), \qquad \partial_t u(0,x)=f_1(x), \qquad  \ldots, \qquad \partial_t^{2n-1}u(0,x)=f_{2n-1}(x)\,. \nonumber
\end{equation}
Consider first $F=0$. In this case one can explicitly solve the Cauchy problem by e.g., taking the Fourier transform of \eqref{eq:toy_n} w.r.t. the $x$ coordinate. Let $\xi$ denote the (spatial) wavenumber and let ${\tilde u}(t,\xi)$ denote the Fourier transform of $u$. Then \eqref{eq:toy_n} with $F=0$ reduces to the ODE
\begin{equation}
	(-\partial_t^2-\xi^2)^n {\tilde u}(t,\xi)=0\,, \nonumber
\end{equation}
whose general solution can be written as
\begin{equation}
	{\tilde u}(t,\xi)=P(t) e^{\ti\xi t}+Q(t) e^{-\ti\xi t}\,,
\end{equation}
where $P$ and $Q$ are both polynomials of degree $n-1$ in the time coordinate $t$. This means that even though a unique solution exists for generic initial data, the solution will grow polynomially in time. This is related to the fact that equation \eqref{eq:toy_n} with $F=0$ is only weakly hyperbolic and the corresponding Cauchy problem is only {\it weakly} well-posed \cite{Kreiss1989}: the $H^s$-norm of the solution can (generically) only be bounded by a higher order Sobolev norm of the initial data. More precisely, assume that the initial data is such that $f_k\in H^{s+n-1-k}$ (where $s\geq 2n+2$). Then for a generic choice of $f_k$, we have
	\begin{equation}
		||u||_{H^s}(t)\leq C(t)\sum\limits_{k=0}^{2n-1}||f_k||_{H^{s+n-1-k}}\,,
	\end{equation}
	for some continuous function $C$.
For some choices of $F\neq 0$ the Cauchy problem is ill-posed. Take for example $n=2$ and $F=-\epsilon\partial_x^3 u$ with a positive constant $\epsilon$. (Of course, this equation is not Lorentz-covariant but it mimics the behaviour of more complicated covariant equations.) In this case one can still write down an explicit solution to \eqref{eq:toy_n} by taking a Fourier transform. The general solution for the Fourier transform ${\tilde u}(t,\xi)$ can be written as
\begin{equation}
		{\tilde u}(t,\xi)=\sum\limits_{s_1,s_2=\pm 1} \tilde{f}_{s_1s_2}(\xi) e^{\lambda_{s_1s_2}(\xi)t}\,,
\end{equation}
where $\lambda_{s_1s_2}(\xi)$ are given by
$$\lambda_{s_1,s_2}=s_1 \sqrt{-\xi^2+s_2 \sqrt{\ti\epsilon \xi^3}}\,.$$
Note that for large wavenumbers ($\xi\to \infty$) the asymptotic behaviour of Re$\lambda_{s_1s_2}$ is $s_1\sqrt{\epsilon|\xi|}$. Therefore, if $s_1=+1$ then ${\tilde u}(t,\xi)$ diverges as $\xi\to\infty $, meaning that there exists no regular solution $u(t,x)$ for any $t>0$ if the initial data is such that $\tilde{f}_{+s_2}\neq 0$.

Despite the previous counterexample, there exist choices of $F$ that are not dangerous. An obvious example is given by $n=2$ and $F=\epsilon^2 \Box u$ with a constant $\epsilon>0$. In this case the solution for the Fourier transform is simply given by
\begin{equation}
	{\tilde u}(t,\xi)=\sum\limits_{s=\pm 1} \left[\tilde{f}_{s}(\xi) e^{s \ti \xi t}+\tilde{g}_{s}(\xi) e^{s \ti \sqrt{\xi^2+\epsilon^2 } t}\right]\,,
\end{equation}
i.e., a sum of a decoupled massless and a massive scalar wave. The Cauchy problem is clearly well-posed and, in fact, the solutions have a better behaviour than in the $F=0$ case since there is no polynomial growth in time.

More generally, let
\begin{equation}
	u^{(p,q)}_{\mu_1\ldots \mu_p}\equiv \partial_{\mu_1}\ldots \partial_{\mu_p}\Box^q u\,,
\end{equation}
and suppose that $F$ is a smooth function of the fields $u^{(p,q)}$ with $p+q\leq n$ and $q\leq n-2$. Then it is possible to recast \eqref{eq:toy_n} as a system of second order wave equations:
\begin{align*}
	\Box u^{(p,q)}_{\mu_1\ldots \mu_p}&=u^{(p,q+1)}_{\mu_1\ldots \mu_p}, \qquad \qquad 0\leq p\leq n-q-1,~ 0\leq q \leq n-2\,, \\
	\Box u^{(n-q,q)}_{\mu_1\ldots \mu_p}&=\partial_{\mu_1} u^{(n-q-1,q+1)}_{\mu_2\ldots \mu_p},\qquad \qquad 0\leq q \leq n-2\,, \\
	\Box u^{(0,n-1)}&=F\,.
\end{align*} 
Hence, the Cauchy problem is locally well-posed for such choices of $F$.

\section{Technical details of the proof of Theorem \ref{thm:short}}

\subsection{Geometric identities}

\begin{lemma}\label{lem:waveWR}
   Let $k,q$ be non-negative integers and $p$ a positive integer. Then the tensor fields $W^{(k)}$ and $R^{(p,q)}$ satisfy geometric identities of the form
   \begin{align}
       \Box W^{(k)}&=F^{(k)}_W\,, \\
       \Box R^{(p,q)}&=R^{(p,q+1)}+F^{(p,q)}_R\,, 
   \end{align}
   such that $F_{(k)}^W$ is a type $(0,k)$ tensor-valued polynomial in the variables $W^{(l)}$ with $0\leq l\leq k$, and $R^{(s,0)}$ with $0\leq s\leq k+2$; $F_{(p,q)}^R$ is a type $(0,p+2)$ tensor-valued polynomial in the variables $W^{(l)}$ with $0\leq l\leq p-1$, and $R^{(s,t)}$ with $s+t\leq p+q$, $t\leq q$. 
\end{lemma}
\begin{proof}
    We will derive the wave equations \eqref{eq:wave_W} and \eqref{eq:wave_R} inductively. For $W^{(k)}$ we do this by induction on $k$. For $k=0$ we have the identity \cite{Choquet-Bruhat:2009xil}
\begin{align}
&\Box R_{abcd} - 2\,R_{aedf}\,R_{b\phantom{e}c}^{\phantom{b}e\phantom{c}f} + 2\,R_{aecf}\,R_{b\phantom{e}d}^{\phantom{b}e\phantom{d}f} + R_{abef}R^{ef}_{\phantom{ef}cd}  
+2\,R_{abe[c}\,R^{e}_{\phantom{e}d]} -2\,\nabla_c\nabla_{[a} R_{b]d} + 2\,\nabla_d\nabla_{[a} R_{b]c} = 0\,.
\label{eq:box_riem}
\end{align}
Writing the Riemann tensor in terms of the Weyl tensor and the Ricci tensor, \eqref{eq:box_riem} expresses $\Box W$ in terms of $R^{(2,0)}$, $R^{(0,0)}$ and $W^{(0)}$, confirming the statement of the lemma for $k=0$. Next, we assume that $W^{(k-1)}$ satisfies a wave equation such that the source term on the RHS is a polynomial in $W^{(l)}$ with $0\leq l\leq k-1$ and $R^{(s,0)}$ with $0\leq s\leq k+1$. To demonstrate that $W^{(k)}$ also satisfies a wave equation of the required form we consider
\begin{align}
\Box \nabla_a W^{(k-1)}&=\nabla_a \Box W^{(k-1)}+\ldots \label{eq:W_k-1} 
\end{align}
where we used the Ricci identity to commute the $\Box$ operator past $\nabla_a$, and the ellipsis stands for terms of the schematic form $W^{(1)}W^{(k-1)}$, $W^{(0)}W^{(k)}$, $R^{(1,0)}W^{(k-1)}$, $R^{(0,0)}W^{(k)}$. It follows from the induction hypothesis that the first term on the RHS of \eqref{eq:W_k-1} can be expressed as a polynomial in the variables $W^{(l)}$ with $0\leq l\leq k$ and $R^{(s,0)}$ with $0\leq s\leq k+2$, closing the induction loop.

Next, we consider the case of $R^{(p,q)}$. Our strategy is now to obtain a wave equation for $R^{(p,q)}$ inductively on $p$. For $p=1$ we have
\begin{align}
    \Box \nabla_c R_{ab}^{(0,q)}=&~\nabla^d\nabla_c\nabla_d R_{ab}^{(0,q)}-\nabla^d\left(R^e{}_{adc}R_{eb}^{(0,q)}+R^e{}_{bdc} R_{ae}^{(0,q)}\right) \nonumber \\
    =&~\nabla^a \Box  R_{ab}^{(0,q)}+R^e{}_c \nabla_e R_{ab}^{(0,q)}-2R^e{}_{adc}\nabla^d R_{eb}^{(0,q)}-2R^e{}_{bdc}\nabla^d R_{ae}^{(0,q)} \nonumber\\
    & -\left(\nabla_a R^e{}_{c}-\nabla^e R_{ac}\right)R_{eb}^{(0,q)}-\left(\nabla_b R^e{}_{c}-\nabla^e R_{bc}\right)R_{ea}^{(0,q)} \,,\label{eq:gradR}
\end{align}
where we used the Ricci identity twice and the contracted Bianchi identity in the second line. We can write \eqref{eq:gradR} schematically as
$$\Box R^{(1,q)}=R^{(1,q+1)}+\text{ terms with $W^{(0)}R^{(1,q)}$, $R^{(1,0)}R^{(0,q)}$, $R^{(0,0)}R^{(1,q)}$}\,, $$
in agreement with the statement of the lemma for $p=1$. Assuming that $R^{(p-1,q)}$ satisfies a wave equation of the required form, we seek a wave equation for $R^{(p,q)}$. To this end, we compute
\begin{equation}
    \Box \nabla_a R^{(p-1,q)}=\nabla_a \Box R^{(p-1,q)}+\ldots \label{eq:R_p-1}
\end{equation}
where we employed the Ricci identity again, producing terms (denoted by the ellipsis) of the schematic form $W^{(1)}R^{(p-1,q)}$, $W^{(0)}R^{(p,q)}$, $R^{(1,0)}R^{(p-1,q)}$, $R^{(0,0)}R^{(p,q)}$. Regarding the structure of the first term in \eqref{eq:R_p-1}, we observe that the induction hypothesis implies that it can be expressed as a sum of $R^{(p,q+1)}$ and a polynomial depending on $W^{(l)}$ with $0\leq l\leq p-1$,
and $R^{(s,t)}$ with $s+t\leq p+q$, $t\leq q$, thus completing the proof.
\end{proof}

\begin{lemma}\label{lem:grad_G}
   Let $n$ be a positive integer. Then $G_{ab}^{(0,n)}$ satisfies an identity of the form
   \begin{equation}
       \nabla^bG_{ab}^{(0,n)}=I^{(n)}_a
   \end{equation}
   where $I^{(n)}_a$ can be expressed as a sum of monomials such that each monomial is a product of factors of the metric and the tensor fields $W^{(k)}$ with $0\leq k \leq n-1$, and $R^{(p,q)}$ with $p+q\leq n$ and $q<n$.
\end{lemma}

\begin{proof}
    Once again, we employ induction on $n$. The first part of the statement can be verified for $n=1$ by the following straightforward calculation
    \begin{align}
     \nabla^a\Box G_{ab}=&~\nabla^c\nabla^a\nabla_c G_{ab}-R^{ad}\nabla_d G_{ab}+R^{cd}\nabla_c G_{db}+R_b{}^{dac}\nabla_c G_{ad} \nonumber \\
     =&~\Box \nabla^a G_{ab}+\nabla^c\left(R^d{}_cG_{db}-R^d{}_b{}^a{}_c G_{ad}\right) 
      -R^{ad}\nabla_d G_{ab}+R^{cd}\nabla_c G_{db}+R_b{}^{dac}\nabla_c G_{ad}\,, \label{eq:gradG}
    \end{align}
    where the Ricci identity was used in two steps. Due to the contracted Bianchi identity, the first term in the second line vanishes and the third term in the second line that involves $\nabla^c R^d{}_b{}^a{}_c$ can be written in terms of $R^{(1,0)}$. It follows that $ \nabla^a\Box G_{ab}$ can be expressed in terms of $R^{(0,0)}$, $R^{(1,0)}$ and $W^{(0)}$ which is exactly the claim of the lemma for $n=1$.

    Next, we assume that $\nabla^bG_{ab}^{(0,n-1)}$ can be expressed in the required form. To obtain the result for $\nabla^bG_{ab}^{(0,n)}$ we calculate (similarly to \eqref{eq:gradG})
    \begin{align}
     \nabla^a\Box G_{ab}^{(0,n-1)}&=\nabla^c\nabla^a\nabla_c G_{ab}^{(0,n-1)}-R^{ad}\nabla_d G_{ab}^{(0,n-1)}+R^{cd}\nabla_c G_{db}^{(0,n-1)}+R_b{}^{dac}\nabla_c G_{ad}^{(0,n-1)} \nonumber\\
     &=\Box \nabla^a G_{ab}^{(0,n-1)}+\ldots \,, \label{eq:gradGnproof}
    \end{align}
    where the ellipsis stands for terms of the schematic form $R^{(0,0)}G^{(1,n-1)}$ and $R^{(1,0)}G^{(0,n-1)}$, i.e., a sum of monomials consisting of factors of $R^{(p,q)}$ with $p+q\leq n$ and $q<n$. Now consider the first term in \eqref{eq:gradGnproof}; using the induction hypothesis and Lemma \ref{lem:waveWR} concludes the proof.

\end{proof}

\subsection{Equation of motion of the regularising theory}

Next we study the structure of the equation of motion of the regularising theory. First, we have the following result.
\begin{lemma}\label{lem:reg_1}
    Let $n$ be a non-negative integer and consider the theory
    \begin{equation}
    S_n=\int \bepsilon\,\ R^{ab}\Box^n G_{ab} \,, \label{eq:S_n}
    \end{equation}
    which is proportional to the action of ${\cal L}_{\text{reg},n}$ with $\alpha_n=2\beta_n$. Variation of this action gives an equation of motion of the following form
    \begin{equation}
    E^{ab}_{(n)}\equiv  \sum\limits_{\lceil n/2 \rceil\leq k\leq n}\left( {\cal P}^{abcd}_{(n,k)}(\nabla) R_{cd}^{(0,k)}+{\cal Q}_{(n,k)}^{abcdef}R_{cd}^{(0,n-k)}R_{ef}^{(0,k)}\right)=0 \,,\label{eq:lem1}
    \end{equation}
    where ${\cal P}^{abcd}_{(n,k)}(\nabla)$ are second order differential operators such that
    \begin{equation}
    {\cal P}_{(n,n)ab}{}^{cd}(\nabla)R_{cd}^{(0,n)}=-\Box G^{(0,n)}_{ab}-g_{ab}\nabla^c\nabla^d G^{(0,n)}_{cd}+2\nabla^c \nabla_{(a}G^{(0,n)}_{b)c} \,,\label{eq:principal}
    \end{equation}
    and for $k<n$
    $${\cal P}_{(n,k)}^{abcd}(\nabla)={\cal A}_{(n,k)}^{abcde_1e_2e_3e_4}R^{(0,n-k-1)}_{e_1e_2} \nabla_{e_3} \nabla_{e_4} + {\cal B}_{(n,k)}^{abcde_1e_2e_3e_4}R^{(1,n-k-1)}_{e_1e_2e_3}\nabla_{e_4}+ {\cal C}_{(n,k)}^{abcde_1e_2e_3e_4}R^{(2,n-k-1)}_{e_1e_2e_3e_4}\,,$$
    and the tensors ${\cal A}_{(n,k)}$, ${\cal B}_{(n,k)}$, ${\cal C}_{(n,k)}$, ${\cal Q}_{(n,k)}$ depend only on the metric.
\end{lemma}

\noindent
{\bf Remark.} It is interesting to note that for a covector $\xi$ the tensor ${\cal P}_{(n,n)}^{abcd}(\xi)$ coincides with the principal symbol of the Einstein equation when written in terms of the metric. This is related to the fact that the theory \eqref{eq:S_n} can be generated from GR by a perturbative field redefinition that involves the lower order equations of motion.

\begin{proof}
    We will prove the statement by induction on $n$. We can explicitly check the validity of the statement for $n=0$. Let $g_{ab}(\lambda)$ be a one-parameter family of metrics and let $\delta g_{ab}\equiv (\ud g_{ab}/\ud \lambda)|_{\lambda=0}$ and similarly for all tensor fields depending on the metric. Variation of \eqref{eq:S_n} with $n=0$ gives
    \begin{equation}
        E^{ab}_{(0)}\equiv {\cal P}^{abcd}_{(0,0)}(\nabla)R_{cd}+2R^{(a}{}_cG^{b)c}-g^{ab}R^{cd}G_{cd}\,,
    \end{equation}
    which is the required structure for $n=0$.
    
    \noindent
    Now suppose that the statement of the lemma holds for the action $S_{n-1}$. We shall show that this implies that the variation of $S_n$ yields an equation with the required structure.

    \noindent
    Consider the variation of $S_{n-1}$ and let $X_{ab}=G^{(0,n-1)}_{ab}$. A calculation gives
    \begin{align*}
        \delta S_{n-1}&=\int \bepsilon\Bigl\{\frac12 R^{cd}X_{cd} g^{ab}\delta g_{ab}+({\cal D}_R\delta g)^{ab}X_{ab}+R^{ab}({\cal D}_X \delta g)_{ab} \Bigr\} \\
        &=\int \bepsilon\Bigl\{\frac12 R^{cd}X_{cd} g^{ab}+{\cal D}_R^\dag X^{ab}+{\cal D}_X^\dag R^{ab}\Bigr\} \delta g_{ab}\,,
    \end{align*}
where ${\cal D}_R$ and ${\cal D}_X$ are differential operators such that $\delta R^{ab}=({\cal D}_R\delta g)^{ab}$ and $\delta X_{ab}=({\cal D}_X\delta g)_{ab}$, ${\cal D}_R^\dag$ and ${\cal D}_X^\dag$ denote the formal adjoints of these differential operators. The equation of motion can then be formally written as
\begin{align}
    E^{ab}_{(n-1)}&=\tfrac12 R^{cd}X_{cd} g^{ab}+{\cal D}_R^\dag X^{ab}+{\cal D}_X^\dag R^{ab} \nonumber \\
    &=\left[\tfrac12 R_{cd} g^{ab}P^{cdef}\Box^{n-1}+P^{abef}{\cal D}_R^\dag \Box^{n-1}+g^{e(a}g^{b)f} {\cal D}_X^\dag \right]R_{ef} \,,\label{eq:E_n-1}
\end{align}
where we substituted $X_{ab}=G^{(0,n-1)}_{ab}$ in the second line and defined
\begin{equation}
    P^{abcd}\equiv g^{a(c}g^{d)b}-\tfrac12 g^{ab}g^{cd}. \label{eq:defP}
\end{equation}
Next we vary $S_n$ as follows
    \begin{align*}
        \delta S_{n}&=\int \bepsilon\Bigl\{\frac12 R^{cd}\Box X_{cd} g^{ab}\delta g_{ab}+({\cal D}_R\delta g)^{ab}\Box X_{ab}+R^{ab}\delta \Box X_{ab} \Bigr\}\,. 
    \end{align*}
The variation of $\Box X$ can be computed as
\begin{equation}
    \delta \Box X_{ab}=\Box \delta X_{ab}-\nabla^c\nabla^d X_{ab}\delta g_{cd}+\ldots \,,\label{eq:varBoxX}
\end{equation}
where the ellipsis stands for terms of the schematic form $X\nabla \delta \Gamma$ and $\nabla X \delta \Gamma$ and $\delta\Gamma^a_{bc}$ denotes the variation of the Christoffel symbols. Integrating by parts to eliminate covariant derivatives acting on $\delta g$ produces terms of the same form as ${\cal P}^{abcd}_{(n,n-1)}(\nabla)R^{(0,n-1)}_{cd}$ in the statement of the lemma. Then we have
    \begin{align*}
             \delta S_{n}&=\int \bepsilon\Bigl\{\frac12 R^{cd}\Box X_{cd} g^{ab}+{\cal D}_R^\dag \Box X^{ab}+{\cal D}_X^\dag \Box R^{ab}+\ldots\Bigr\} \delta g_{ab}\,,
    \end{align*}
and thus
\begin{align}
        E^{ab}_{(n)}&=\frac12 R^{cd}\Box X_{cd} g^{ab}+{\cal D}_R^\dag \Box X^{ab}+{\cal D}_X^\dag \Box R^{ab}+\ldots \nonumber \\
        &=\left[\frac12 R_{cd} g^{ab}P^{cdef}\Box^{n-1}+P^{abef}{\cal D}_R^\dag \Box^{n-1}+g^{a(e}g^{b)f}{\cal D}_X^\dag \right]R_{ef}^{(0,1)}+\ldots \,.\label{eq:E_n}
\end{align}
Here the ellipsis stands for terms with the same structure as ${\cal P}^{abcd}_{(n,n-1)}(\nabla)R^{(0,n-1)}_{cd}$ as discussed above. Note that the differential operator in the square brackets acting on $R^{(0,1)}$ is the same as the differential operator acting on $R$ in \eqref{eq:E_n-1}. It follows from the induction hypothesis that the terms explicitly displayed in \eqref{eq:E_n}  have the desired form with ${\cal P}^{abcd}_{(n,k)}={\cal P}^{abcd}_{(n-1,k)}$. This concludes the proof.
\end{proof}

\begin{lemma}\label{lem:scalar}
    Let $n$ be a non-negative integer and consider the theory
    \begin{equation}
    {\bar S}_n=\int \bepsilon\,\ R\Box^n R \,.\label{eq:scalarS_n}
    \end{equation}
    Variation of this action gives an equation of motion of the following form
    \begin{equation}
    {\bar E}^{ab}_{(n)}\equiv  \sum\limits_{\lceil n/2 \rceil\leq k\leq n}\left( \bar{\cal P}^{ab}_{(n,k)}(\nabla) {\bar R}^{(0,k)}+\bar{\cal Q}_{(n,k)}^{abcd}{ R}^{(0,n-k)}_{cd}{\bar R}^{(0,k)}\right)=0 \,,\label{eq:lem_scalar}
    \end{equation}
    where $\bar{\cal P}^{ab}_{(n,k)}(\nabla)$ are second order differential operators such that
    \begin{equation}
    \bar{\cal P}_{(n,n)ab}{}(\nabla)=2\left(\nabla^a\nabla^b-g^{ab} \Box\right) \,, \label{eq:sc_principal}
    \end{equation}
    and for $k<n$
    $$\bar{\cal P}_{(n,k)}^{ab}(\nabla)=\bar{\cal A}_{(n,k)}^{abcd}{\bar R}^{(0,n-k-1)} \nabla_c \nabla_d + \bar{\cal B}_{(n,k)}^{abcd}\bar{R}^{(1,n-k-1)}_{c}\nabla_d+ \bar{\cal C}_{(n,k)}^{abcd}\bar{R}^{(2,n-k-1)}_{cd}\,,$$
    and the tensors $\bar{\cal A}_{(n,k)}$, $\bar{\cal B}_{(n,k)}$, $\bar{\cal C}_{(n,k)}$, $\bar{\cal Q}_{(n,k)}$ depend only on the metric.
\end{lemma}

\noindent
{\bf Remark.} Lemmas \ref{lem:reg_1}--\ref{lem:scalar} imply that any Ricci flat metric is a solution of the regularising theory.

\begin{proof}
    The proof is very similar to that of Lemma \ref{lem:reg_1}. We employ induction on $n$. Variation of \eqref{eq:scalarS_n} with $n=0$ gives
    \begin{equation}
        {\bar E}^{ab}\equiv 2\,\nabla^a\nabla^b R-2\,g^{ab} \Box R+R^2 \,g^{ab}-2\,R\, R^{ab}\,,
    \end{equation}
    which is the required structure for $n=0$.
    
    \noindent
    Next, assume that the equation of motion of ${\bar S}_{n-1}$ has the required structure. Varying ${\bar S}_{n-1}$ (and letting $Y\equiv{\bar R}^{(0,n-1)}$) gives
    \begin{align*}
        \delta {\bar S}_{n-1}&=\int \bepsilon\Bigl\{\frac12 R \,Y g^{ab}\delta g_{ab}+({\cal D}_R\delta g) Y+R({\cal D}_Y \delta g) \Bigr\} \\
        &=\int \bepsilon\Bigl\{\frac12 R\, Y g^{ab}+({\cal D}_{\bar R}^\dag Y)^{ab}+({\cal D}_Y^\dag R)^{ab}\Bigr\} \delta g_{ab}\,,
    \end{align*}
where ${\cal D}_{\bar R}$ and ${\cal D}_Y$ are differential operators such that $\delta R=({\cal D}_{\bar R})^{ab}\delta g_{ab}$ and $\delta Y=({\cal D}_Y)^{ab}\delta g_{ab}$, ${\cal D}_R^\dag$ and ${\cal D}_Y^\dag$ denotes the formal adjoints of these differential operators. The equation of motion can then be formally written as
\begin{align}
    {\bar E}^{ab}_{(n-1)}&=\frac12 R\, Y g^{ab}+({\cal D}_{\bar R}^\dag Y)^{ab}+({\cal D}_Y^\dag R)^{ab} \nonumber \\
    &=\left[\frac12 R g^{ab}\Box^{n-1}+({\cal D}_{\bar R}^\dag)^{ab} \Box^{n-1}+ ({\cal D}_Y^\dag)^{ab} \right]R\,. \label{eq:barE_n-1}
\end{align}
By the induction hypothesis this equation has the form described in the statement of the lemma. Next we vary $S_n$:
    \begin{align*}
        \delta S_{n}&=\int \bepsilon\Bigl\{\frac12 R\Box Y g^{ab}\delta g_{ab}+({\cal D}_{\bar R})^{ab}\delta g_{ab}\Box Y+R\delta \Box Y \Bigr\} \,.
    \end{align*}
The variation of $\Box Y$ is given by
\begin{equation}
    \delta \Box Y=\Box \delta Y-\text{ terms with $Y\nabla \delta \Gamma$ and $\nabla Y \delta \Gamma$} \,.
\end{equation}
Integrating by parts to transfer derivatives from $\delta g$ to $R$ and $Y$ produces terms of the same form as $\bar{\cal P}^{ab}_{(n,n-1)}(\nabla){\bar R}^{(0,n-1)}$ in the statement of the lemma. Then we have
    \begin{align*}
             \delta {\bar S}_{n}&=\int \bepsilon \Bigl\{\frac12 R\Box Y g^{ab}+({\cal D}_{\bar R}^\dag)^{ab} \Box Y+({\cal D}_Y^\dag)^{ab} \Box R+\ldots\Bigr\} \delta g_{ab}\,,
    \end{align*}
and thus
\begin{align}
        {\bar E}^{ab}_{(n)}&=\frac12 R\Box Y g^{ab}+({\cal D}_{\bar R}^\dag)^{ab} \Box Y+({\cal D}_Y^\dag)^{ab} \Box R+\ldots \nonumber \\
        &=\left[\frac12 R g^{ab}\Box^{n-1}+({\cal D}_{\bar R}^\dag) \Box^{n-1}+({\cal D}_Y^\dag)^{ab} \right]{\bar R}^{(0,1)}+\ldots \,.
\end{align}
The terms suppressed by the ellipsis are the terms with the structure of $\bar{\cal P}^{ab}_{(n,n-1)}(\nabla){\bar R}^{(0,n-1)}$ discussed above. Similarly to the proof of Lemma \ref{lem:reg_1} we notice that the differential operator in the square brackets acting on $\bar{R}^{(0,1)}$ is the same as the differential operator acting on $R$ in \eqref{eq:barE_n-1} and hence the explicitly displayed term has the required structure, closing the induction loop.
\end{proof}

\subsection{Constraint propagation}

To complete the proof of the well-posedness part of Theorem \ref{thm:short}, we provide details below on how to obtain a system of homogeneous linear wave equations for the constraint variables. 

The contracted Bianchi identity implies (see e.g., \cite{Wald:1984rg})
    \begin{equation}
        g^{\alpha\beta}\partial_\alpha \partial_\beta \Gamma^\mu+\mathbb{L}(g,\partial)^\mu{}_\nu\Gamma^\nu=2{\cal B}^{(0)\mu} \,,
    \end{equation}
    where $\mathbb{L}$ is a first order differential operator depending on the metric and its partial derivatives up to second order. 
    
    \noindent
    To obtain evolution equations for ${\cal R}^{(p,q)}$ and ${\cal W}^{(k)}$, consider the wave equations for ${R}^{(p,q)}$ and ${W}^{(k)}$ derived in Lemma \ref{lem:waveWR}. On the one hand, these equations are part of the evolution system \eqref{eq:reduced_wave} when written in terms of the fields ${R}^{(p,q)}$ and ${W}^{(k)}$. On the other hand, these equations also hold as geometric identities when expressed in terms of derivatives of the Ricci tensor and the Weyl tensor. Now take the difference of the respective evolution equation and geometric identity for each of ${R}^{(p,q)}$, ${W}^{(k)}$. These equations will have the following form:
    \begin{align}
        \Box {\cal R}^{(p,q)}&={\cal F}_R^{(p,q)}\,, \\
        \Box {\cal W}^{(k)}&={\cal F}_W^{(k)}\,,
    \end{align}
    where ${\cal F}_R^{(p,q)}$ and ${\cal F}_W^{(k)}$ can be written as a sum of terms of the form
    \begin{equation}
      M_i(\mathcal{T})-M_i(\mathcal{S})\,, \label{eq:monomial_diff} 
    \end{equation}
    where $M_i$ is a monomial built from factors of its arguments (and contractions with respect to the metric), ${\cal T}$ stands for the tensor fields $R^{(p,q)}$, $W^{(k)}$ and ${\cal S}$ stands for the tensor fields $\nabla_{c_1}\ldots \nabla_{c_p} \Box^q R_{ab}$, $\nabla_{e_1}\ldots \nabla_{e_k} W_{abcd}$. We claim that any expression of the form \eqref{eq:monomial_diff} can be written as
    \begin{equation}
    M_i({\cal T})-M_i({\cal S})=\sum\limits_k A_{i,k}({\cal T},{\cal S}) \left({\cal T}_k-{\cal S}_k\right) \,. \label{eq:monomial_diff_2}
    \end{equation}    
    (Note that $\left({\cal T}_k-{\cal S}_k\right)$ is equal to one of ${\cal R}^{(p,q)}$ or ${\cal W}^{(k)}$.) This claim can be verfied by e.g., using induction on the number of factors in $M_i$. For a monomial containing one factor, the statement is obviously true. Suppose that \eqref{eq:monomial_diff_2} holds for any monomial built from a product of $r>1$ factors. Now consider a monomial $M$ built from a product of $r+1$ factors. Then we have
    \begin{align}
        M({\cal T})-M(\cal S)&={\cal T}_k N_k({\cal T})-{\cal S}_k N_k({\cal S}) \nonumber \\
        &=\left({\cal T}_k-{\cal S}_k\right) N_k({\cal T})+\left(N_k({\cal T})- N_k({\cal S})\right){\cal S}_k \,, \label{eq:monomial_ind}
    \end{align}
    where $N_k$ is a monomial expressed as a product of $r$ factors. Due to the induction hypothesis the second term in \eqref{eq:monomial_ind} can be expressed in the form \eqref{eq:monomial_diff_2}, thus closing the induction loop. It follows that ${\cal F}_R^{(p,q)}$ and ${\cal F}_W^{(k)}$ can be written as homogeneous linear expressions in ${\cal R}^{(p,q)}$ and ${\cal W}^{(k)}$.

    \noindent
    Finally, we consider the propagation equations for the Bianchi constraints. First, we act with a $\Box$ operator on both sides of equations \eqref{eq:Bianchi} (for each $l$). Next, we commute the $\Box$ operator past the covariant derivatives as explained in the proof of Lemma \ref{lem:grad_G}. For $l<n$ this simply gives (noting that $\Box I_\mu^{(l)}=I_\mu^{(l+1)}$ due to the evolution equations for $W^{(k)}$ and $R^{(p,q)}$)
    \begin{align}
        \Box {\cal B}_\mu^{(l)}&={\cal B}_\mu^{(l+1)}\,.
    \end{align}
    For ${\cal B}^{(n)}$, we note that the equation of motion $E_{\text{aug}}^{ab}$ of the augmented theory (obtained by varying the action) satisfies the off-shell identity
    \begin{equation}
        \nabla_\nu E^{\mu\nu}_{\text{aug}}=0 \,,\label{eq:diff_id}
    \end{equation}
    since $E_{\text{aug}}$ comes from a diffeomorphism-invariant action. Let
    \begin{equation}
        \mathbb{E}_{\mu\nu}\equiv -\Box R^{(0,n)}_{\mu\nu}+F_{\mu\nu}\,,
    \end{equation}
    denote the evolution equation for $R^{(0,n)}$ used in the augmented theory. This equation is related to $E_{\text{aug}}$ through
    \begin{equation}
        E_{\text{aug}}^{\mu\nu}+{\cal S}^{\mu\nu}=P^{\mu\nu\alpha\beta}\left(\mathbb{E}_{\alpha\beta}-\nabla_\alpha {\cal B}^{(l)}_\beta\right)\,, \label{eq:aug_to_ev}
    \end{equation}
    where $P$ was defined in \eqref{eq:defP} and ${\cal S}^{\mu\nu}$ is a sum of monomials such that each monomial is a product of factors of the metric, $\bepsilon$ (in parity-violating theories), $R^{(p,q)}$, $W^{(k)}$ and exactly one factor of a constraint variable such that (c.f. the argument following \eqref{eq:monomial_diff})
    \begin{equation}
        E_{\text{aug}}^{\mu\nu}[R,W, \nabla R, \nabla W, \ldots]+{\cal S}^{\mu\nu}=E_{\text{aug.}}^{\mu\nu}[R^{(p,q)},W^{(k)}]\,, \label{eq:aug_convert}
    \end{equation}
    i.e., the addition of ${\cal S}_{\mu\nu}$ converts $E_{\text{aug}}^{\mu\nu}$ written in terms of the curvature tensors and their derivatives to $E_{\text{aug}}^{\mu\nu}$ written in terms of the dynamical variables $R^{(p,q)}$, $W^{(k)}$. Taking a gradient of \eqref{eq:aug_to_ev} and using \eqref{eq:diff_id} then gives
    \begin{equation}
        \Box {\cal B}^{(n)}_\mu+R_\mu{}^\nu {\cal B}^{(n)}_\nu={\cal F}_\mu^{(n)}\,,
    \end{equation}
where ${\cal F}_\mu^{(n)}=-\nabla^\nu{\cal S}_{\mu\nu}$ is a homogeneous linear expression in the constraint variables. This concludes the derivation of the evolution equations for the constraints.

\end{document}